\documentclass[letterpaper]{article} 
\usepackage{aaai25}  
\usepackage{times}  
\usepackage{helvet}  
\usepackage{courier}  
\usepackage[hyphens]{url}  
\usepackage{graphicx} 
\usepackage{natbib}  
\usepackage{caption} 
\usepackage{algorithm}
\usepackage{algorithmic}

\usepackage{newfloat}
\usepackage{listings}
\DeclareCaptionStyle{ruled}{labelfont=normalfont,labelsep=colon,strut=off} 
\floatstyle{ruled}
\newfloat{listing}{tb}{lst}{}
\floatname{listing}{Listing}
\usepackage{amsmath}
\usepackage{amsthm}
\usepackage{booktabs, makecell}

\usepackage{amssymb} 
\usepackage{amsmath}
\usepackage{enumitem}
\usepackage{nicefrac}
\usepackage{tikz}
\usetikzlibrary{arrows,shapes, calc, fit, positioning,patterns}
\usepackage{tkz-graph}
\usepackage{comment}
\usepackage{algorithm}
\usepackage{algorithmic}

\usepackage{thmtools} 
\usepackage{thm-restate}
\usepackage{soul}
\usepackage{booktabs}
\usepackage{multirow}

\usepackage{pifont}

\usepackage{caption}
\usepackage{subcaption}

\usepackage[]{todonotes}
\definecolor{auburn}{rgb}{0.43, 0.21, 0.1}
\newcommand{\commentLaurent}[1]{\textcolor{red}{(L: #1)}}
\newcommand{\commentAyumi}[1]{\textcolor{blue}{(Ayumi: #1)}}

\newcommand{\commentLuca}[1]{{\bf\textcolor{auburn}{(Luca: #1)}}}

\newtheorem{theorem}{Theorem}[section]
\newtheorem{cor}[theorem]{Corollary}
\newtheorem{lemma}[theorem]{Lemma}
\newtheorem{prop}[theorem]{Proposition}
\newtheorem{claim}[theorem]{Claim}

\newtheorem{define}[theorem]{Definition}

\newtheorem{example}[theorem]{Example}

\newcommand{\cmark}{\ding{51}}%
\newcommand{\xmark}{\ding{55}}%

\newcommand{\bbN}{\mathbb{N}}

\newcommand{\hg}{HG}

\newcommand{\calF}{\mathcal{F}}

\def\x{{X}}

\def\plus{\text{+}}

\title{
Individually Stable Dynamics in Coalition Formation over Graphs}
\author{
Angelo Fanelli\textsuperscript{\rm 1}, 
Laurent Gourv\`{e}s\textsuperscript{\rm 1},  
Ayumi Igarashi\textsuperscript{\rm 2}, 
Luca Moscardelli\textsuperscript{\rm 3}
}
\affiliations{
    \textsuperscript{\rm 1}LAMSADE – 
    Université Paris Dauphine-PSL, France\\
     \textsuperscript{\rm 2}The University of Tokyo, Japan\\
     \textsuperscript{\rm 3}University of Chieti-Pescara, Italy

}

\begin{document}
\maketitle

\begin{abstract}

Coalition formation over graphs is a well studied class of games whose players are vertices and feasible coalitions must be connected subgraphs.  
In this setting, the existence and computation of equilibria, under various notions of stability, has attracted a lot of attention. However, the natural process by which players, starting from any feasible state, strive to reach an equilibrium after a series of unilateral improving deviations, has been less studied. 
We investigate the convergence of dynamics towards individually stable outcomes under the following perspective: what are the most general classes of preferences and graph topologies guaranteeing convergence? 
To this aim, on the one hand, we cover a hierarchy of preferences, ranging from the most general to a subcase of additively separable preferences, including individually rational and monotone cases. On the other hand, given that convergence may fail in graphs admitting a cycle even in our most restrictive preference class, we analyze acyclic graph topologies such as trees, paths, and stars.

\end{abstract}

\section{Introduction}


Coalition formation is an important and widely investigated
issue in 
artificial intelligence.  
In 
many economic, social, and political settings, individuals carry out activities in groups rather than by themselves.  
\emph{Hedonic games}, introduced by \citet{Dreze1980} and later developed in \cite{Banerjee2001,Bogomolnaia2002,romero2001stability}, are among the most important game-theoretic approaches to the study of coalition formation problems. 
An outcome for these games is 
a partition of the players into coalitions, over which the players have utilities. 
A player's utility 
depends on the coalition she belongs to, and is not affected by how 
the participants of other coalitions are partitioned.

The standard model of hedonic games does not impose any restrictions on which coalitions may form. However, in reality, we often encounter {\em network} constraints on coalition formation, i.e.,    
entities can communicate and cooperate only if there are connected. 
Such restrictions 
can be naturally described by means of undirected graphs, giving life to \emph{graph hedonic games}, in which players 
are identified with vertices, 
communication links with edges, and feasible coalitions with connected subgraphs \cite{Myerson77,Demange2004,igarashi2016hedonicgraph}.

In this paper, we study graph hedonic games under the perspective of {\em individual stability} (IS for short), a natural notion of stability introduced by \citet{Dreze1980}. 
In particular, a player $i$ performs an 
{\em IS deviation} to coalition $T$ whenever $i$ 
prefers coalition $T \cup \{i\}$ to her current coalition, and also all players in $T$ 
do not prefer $T$ to $T \cup \{i\}$. 
Roughly speaking, an IS deviation is a Nash deviation 
with the additional constraint that all players in coalition $T$ have to ``accept'' player $i$. A partition 
is individually stable if no player has an IS 
deviation available. 
Real life examples of IS deviations exist: 
a nation can enter the NATO, or the EU, 
only if its members unanimously agree.



Note that stability is mostly 
concerned with the final state of the coalition formation process and one frequently ignores how these desirable partitions can actually be reached. 
Essentially, speaking about stability often implicitly assumes that there is a central authority knowing the preferences of all players and computing a stable partition. 
However, the existence of a stable state does not imply that players, starting from an initial configuration, 
eventually reach stability after a finite sequence of 
improving deviations. In the well known stable marriage problem, 
a stable matching always exists under mild assumptions, 
and the centralized algorithm of \citet{GaleShapley} outputs such a state. However, the natural process of sequentially eliminating the existing blocking pairs may loop \cite{knuth1997stable}. Examples of this kind motivate the study of 
decentralized coalition formation processes operated by autonomous entities. 

\begin{table*}[ht!]
	\footnotesize
	\centering
	\begin{tabular}{l|ccccc}
		\toprule
		 & general & individually rational (IR)& monotone & LAS\\
		\midrule
		cycles & \xmark & \xmark& \xmark & \xmark~ (Ex.~\ref{counterexample:cycles})\\
		trees & \xmark  & \xmark  & \xmark~ (Ex.~\ref{counterexample:mon:tree}) & \cmark ~(Th.~\ref{thm:trees}) \\
		paths & \xmark~(Ex.~\ref{ex:2coalitions}) & \xmark~ (Ex.~\ref{counterexample:paths} for $\ge 4$ coalitions) & \cmark ~(Th.~\ref{thm:paths:monotone}) & \cmark ~(Th.~\ref{thm:paths:monotone})\\
		 &  & \cmark~(Th.~\ref{thm:three:paths} for $\leq 3$ coalitions) & & \\
        stars & \xmark ~(Ex.~\ref{counterexample:star}) & \cmark ~(Th.~\ref{IS:stars}) & \cmark~ (Th.~\ref{IS:stars}) & \cmark ~(Th.~\ref{IS:stars})\\
		\bottomrule
	\end{tabular}
\caption{
Overview of the convergence results. The ``\xmark" indicates that the dynamics are not guaranteed to converge, while ``\cmark" indicates that the dynamics under individual stability notion converge. The existence of an individually stable partition is always guaranteed except when the graph contains a cycle and the preferences are general \citep{igarashi2016hedonicgraph}. Ex.~\ref{ex:2coalitions} is in the appendix.  
}
\label{table:HGT}
\end{table*}

This article focuses on IS dynamics. We  
analyze the process in which the players interact and, possibly,  
reach an individually stable coalition structure. 
More precisely, we aim at determining the most general possible classes of preferences and graph topologies guaranteeing convergence of IS dynamics in  graph hedonic games (for any instance and 
initial state). 
To this aim, we consider the following natural classes of preferences, that we list from the most general to the most particular: 
\emph{general} preferences 
(no restrictions except the ones that are usual in hedonic games), 
\emph{individually rational} (IR for short) preferences (each player likes any coalition she belongs to at least as much as if she were alone), 
\emph{monotone} preferences (each player likes any superset of any coalition $S$ she belongs to at least as much as $S$), and \emph{local additively separable} (LAS for short) preferences. LAS preferences specialize the well studied case of \emph{additively separable} (AS for short) preferences where each player $i$ assigns a value $v_i(j)$ to any other player $j$, and a player prefers the coalition maximizing the sum of values towards its members.  
In LAS preferences, which apply to graph hedonic games, $v_i(j)$ is non negative, and it can be strictly positive only if $i$ and $j$ are neighbours in the graph.


\medskip

\noindent {\bf Contributions.}  
Our results 
are summarized in Table \ref{table:HGT}. 
On the negative side, 
general preferences do not 
guarantee convergence even for the very basic topologies of paths and stars. Analogously, when considering graphs containing cycles, even the most particular class of preferences under study, that is LAS preferences, does not ensure convergence. 
On the positive side, while IR 
preferences 
ensure 
convergence for stars, monotone preferences are needed in order to secure convergence on paths\footnote{When starting from  partitions 
with at most $3$ coalitions, individually rational preferences also  guarantee convergence for paths.} and LAS preferences are needed for guaranteeing convergence on trees. For LAS preferences or when the graph is a star,
we also provide worst case bounds on 
the number of steps that the dynamics need before reaching a stable outcome.

\medskip

\noindent {\bf Related work.}  The book chapter by \citet{AS16} gives a comprehensive overview of 
hedonic games ({\hg}s for short); see also \cite{JH06} for an earlier survey on coalition formation games. Various notions of stability have been considered for {\hg}s, such as core,\footnote{When considering \emph{core stability}, an outcome is stable if and only if there exists no coalition that could make all its members 
better off. 
} 
Nash, or  individual stability. These concepts refer to states that exclude a certain kind of deviations, so they are 
naturally related to the dynamics where, starting from some initial partition, the corresponding deviation is repeatedly performed if it is possible. When it converges,  such a decentralized process constitutes a  plausible explanation of the formation of coalitions.     

One of the best known settings is that of {\hg}s with  
{\em symmetric} AS preferences ($v_i(j)=v_j(i)$ for every pair $i,j$ of players) 
for which the IS dynamics always converge \cite{Bogomolnaia2002}. In fact, in such games, the sum of all players' utilities always increases after the deviation of a player to her preferred coalition. \citet{Gairing2010}, however, showed that finding an IS 
state is PLS-complete, meaning that it is unlikely that the IS dynamics converge after a polynomial number of steps.

In a significant part of the literature on 
cooperative games, 
some graph constraints are imposed on the coalitions that can form. 
The classical result of \citet{Demange2004} in non-transferable utility games translates to the fact that a core stable outcome   
always exists in a hedonic game on a tree. \citet{igarashi2016hedonicgraph} strengthened the result and showed that there is an outcome that satisfies both core and individual stability; 
further, they provided a polynomial time algorithm to compute an individually stable outcome whenever the graph is a tree.  

Concerning the existence of IS states, note that when the graph is connected, having  IR 
preferences implies that the 
state where all players are in the same coalition 
is IS. For general preferences, an IS 
partition is guaranteed to exist when the graph is a forest while it may not exist when the graph contains a cycle 
\cite{igarashi2016hedonicgraph}.

There is a growing literature on the dynamics in {\hg}s. 
\citet{Brandt_Bullinger_Tappe_2022} studied the convergence of dynamics associated with relaxed notions of IS deviations where, instead of requiring unanimous consents, the deviation of a  player is doable if it is accepted by 
a majority of the members of the welcoming coalition.   In another closely related work, \citet{Brandt2023} 
studied the convergence of IS dynamics in 
anonymous {\hg}s, hedonic diversity games, fractional {\hg}s, and dichotomous {\hg}s.  
Their results do not readily compare to Table \ref{table:HGT} since their games 
don't have 
graph connectivity constraints 
and their 
players' preferences are very different. For instance, individual rationality is not necessarily satisfied in hedonic diversity games, fractional {\hg}s, and dichotomous {\hg}s. \citet{Boehmer_Bullinger_Kerkmann_2023} studied the dynamics in {\hg}s where the utilities of players change over time, depending on the history of the coalition formation process. \citet{Hoefer2018} considered dynamics towards  
core stable states in a general hedonic coalition formation game with various constraints of visibility and externality, 
where the players have {\em correlated preferences}, i.e., all members of a coalition have the same utility. 

\medskip 
\noindent {\bf Organization.} The remainder of this article is organised as follows. 
Our model, concepts and dynamics are formally defined in Section \ref{sec:model}. The subsequent sections are dedicated to specific graph topologies: paths (Section \ref{sec:paths}), stars (Section \ref{sec:stars}), and trees (Section \ref{sec:trees}). We finally conclude with directions for future work. Due to space constraints, missing elements such as proofs are placed in the supplementary material.

\section{Model} \label{sec:model}


A {\it graph hedonic  game} 
is defined on a finite graph  $(N,L)$ where $N$ is a set of $n \geq 2$ players, and $L$ is a set of undirected edges between players. 
Players are able to cooperate if and only if they are connected in the graph $(N,L)$. Let $\calF$ be the set of all nonempty subsets $S$ of $N$ such that the subgraph induced by $S$ is connected.
Each player $i \in N$ has a preference ordering $\succeq_{i}$ over the subsets in $\calF(i):=\{\, S \in \calF \mid  i\in S \,\}$. The subsets of $N$ are referred to as {\em coalitions}. A coalition is said to be {\em feasible} if it belongs to $\calF$.
A partition $\pi$ of $N$ is said to be {\em feasible} if $\pi \subseteq \calF$. 
An {\em outcome} or {\em state} (we interchangeably use these terms) 
of a graph hedonic game is a feasible partition. 
For player $i \in N$ and a partition $\pi$ of $N$, we denote by $\pi(i)$ the coalition to which $i$ belongs. We assume without loss of generality that $(N,L)$ is connected (otherwise each connected component can be treated separately).

\subsection{Preferences}
Fix any player $i \in N$, and $S,T \in \calF(i)$. $S \succeq_i T$ means that $S$ is at least as good as $T$ from player $i$'s viewpoint. 
We write $S \succ_i T$ to express that $i$ strictly prefers $S$ over $T$, whereas $S \sim_i T$ means that $i$ is indifferent between $S$ and $T$. 

As is standard in hedonic games, we always assume that $\succeq_{i}$ is complete, transitive, and reflexive \cite{AS16}.
The preference relation $\succeq_i$ is said to be \emph{general} if no further assumption is made on it. 
The preference relation $\succeq_i$ is said to be \emph{individually rational} (IR) if $S \succeq_i \{i\}$ holds for every $S \in \calF(i)$. 
The preference relation $\succeq_i$ is said to be {\em monotone} if $S \succeq_i T$ always holds when $T \subseteq S$. 

In the well studied case of {\em additively separable preferences}, every player $i$ has a value $v_i(j)$ for being in the same coalition as player $j$. Player $i$ has \emph{utility} $\sum_{j \in S \setminus \{i\}} v_i(j)$ when she is in coalition $S$. 
The utility is 0 when a player is alone. 
Then, $S \succeq_i T$ holds when $\sum_{j \in S \setminus \{i\}} v_i(j) \ge \sum_{j \in T \setminus \{i\}} v_i(j)$. In this article we also consider \emph{local additively separable} (LAS in short) preferences, a special case of additively separable preferences 
where the players' values are  non symmetric ($v_i(j)$  can differ from $v_j(i)$), 
non negative, and $v_i(j)$ can be positive only if $i$ and $j$ are neighbors in $(N,L)$, i.e., $v_i(j)>0 \Rightarrow (i,j) \in L$.

The hierarchy of the above preference relations is \begin{center}General $\supseteq$ Individually Rational $\supseteq$ Monotone $\supseteq$ LAS\end{center}   
since LAS preferences are monotone, and monotone preferences are IR.  
An instance  of the graph hedonic game is said to be ${\cal P}$ with ${\cal P} \in\{$general, individually rational, monotone, LAS$\}$ when the preferences of {\em all} the players are ${\cal P}$. 

\subsection{The dynamics under individual stability}

Consider a player $i \in N$, a feasible coalition $S$ containing $i$, and $T \subseteq N \setminus S$. 
A player $i$ {\it wants to deviate} from $S$ to $T$ if $T \cup \{i\} \in \calF$ and $T \cup \{i\} \succ_{i} S$. 
A player $j \in T$ {\it accepts} a deviation of $i$ to $T$ if $T\cup \{i\} \succeq_{j} T$. Thus, an {\em IS deviation} by $i$ from $S$ to $T$ 
is possible if $i$ wants it, 
and all players in $T$ accept it. As a result of player $i$'s deviation, $S$ and $T$ become $S \setminus \{i\}$ and $T \cup \{i\}$, respectively.


\begin{define}
A feasible partition $\pi$ of $N$ is said to be {\it individually stable} (IS) if no player $i \in N$ has an IS 
deviation to another coalition in $\pi \cup \{\emptyset\}$. 
\end{define}

This article focuses on the dynamics associated with individual stability (a.k.a. IS dynamics) as described in Algorithm \ref{alg:is:general}. The deviations are sequential 
and players keep on deviating if the current partition is not individually stable. If, in the dynamics, several players are eligible for an IS deviation, then we suppose that one of them is chosen arbitrarily.

\begin{algorithm}                      
\caption{IS dynamics}         
\label{alg:is:general}                          
\begin{algorithmic}[1]                  
\REQUIRE a graph hedonic game $(N,L,(\succeq_i)_{i\in N})$ and an initial feasible partition $\pi_0$
\ENSURE A feasible partition $\pi$ of $N$ 
\STATE $\pi \leftarrow \pi_0$
\WHILE{
there exists an IS deviation of  
$i\in N$ from $\pi(i)$ to $T\in \pi \cup \{\emptyset\}$
}
\STATE  $\pi \leftarrow (\pi \setminus \{\pi(i),T\})\cup \{T\cup \{i\}\} \cup \{\, S \mid S~\mbox{is a maximal connected subset of}~\pi(i)\setminus \{i\}\,\}$.\label{line:ISdeviation}
\ENDWHILE
\RETURN $\pi$
\end{algorithmic}
\end{algorithm}


When a player $i$ leaves her coalition $\pi(i) \in \calF(i)$, the graph induced by $\pi(i) \setminus \{i\}$ is not necessarily connected. In the IS dynamics, 
it is assumed that the members of $\pi(i) \setminus \{i\}$ reconfigure themselves in a minimum number of feasible coalitions by forming inclusionwise maximal connected subsets of $\pi(i) \setminus \{i\}$ (cf. Line \ref{line:ISdeviation}). The motivation behind this assumption is to consider the minimal changes in $\pi(i)$ caused by  the departure of $i$.


The IS dynamics consist of successive 
{\em better moves}, 
i.e., a player does not necessarily join her most preferred coalition, within the set of coalitions that would accept her. 


A {\em sequence} in the IS dynamics is an ordered list of states $\langle \pi_0, \ldots,\pi_k \rangle$ where each $\pi_t$ is obtained from $\pi_{t-1}$ by a single IS deviation. The sequence $\langle \pi_0, \ldots,\pi_k \rangle$ is {\em cyclic} (equivalently, $\langle \pi_0, \ldots,\pi_k \rangle$ is a {\em cycle}) if $\pi_0=\pi_k$. We say that the IS dynamics {\em converge} if its input does not admit any cyclic sequence. Otherwise, we say that the IS dynamics {\em cycle}. 

The IS dynamics 
can cycle in a graph 
which is a cycle,  
even if the preferences are LAS (cf. Example \ref{counterexample:cycles}).

\begin{example}\label{counterexample:cycles}
In this instance $N=\{a,b,c\}$, $L=\{(a,b),(b,c),(a,c)\}$, $v_a(b)=v_b(c)=v_c(a)=1$, and any other value is 0, so the preferences are: 
\begin{itemize}
\item
a: $\{a,b,c\}  \sim \{a,b\} \succ \{a,c\} \sim \{a\} $ 
\item
b: $\{a,b,c\} \sim \{b,c\} \succ \{a,b\} 
\sim  \{b\} 
$ 
\item
c: $\{a,b,c\} \sim \{a,c\}  \succ \{b,c\}
\sim \{c\}
$ 
\end{itemize}

The initial partition $\pi_0$ is $\{\{a,c\},\{b\}\}$. Player $a$ moves from $\{a,c\}$ to $\{b\}$ giving $\{\{a,b\},\{c\}\}$. Player $b$ moves from $\{a,b\}$ to $\{c\}$ giving $\{\{a\},\{b,c\}\}$. Player $c$ moves from $\{b,c\}$ to $\{a\}$ giving $\{\{a,c\},\{b\}\}$,  which is $\pi_0$.
\end{example}

Example \ref{counterexample:cycles} relies on a cycle of 3 vertices but a similar non-convergence result  can be shown when $(N,L)$ is a cycle with $n > 3$ vertices (cf. Appendix).
Example \ref{counterexample:cycles} tells us that   
even if we consider our most restricted type of preferences (i.e., LAS preferences), convergence may fail if the underlying graph contains a cycle. 
For this reason, in the remainder of this article, we consider acyclic graphs, i.e., paths, stars, and trees.  

\section{Paths} \label{sec:paths}
This section is devoted to the case in which $(N,L)$ is a path. To this respect, we provide a complete picture of the convergence of IS dynamics.
In particular, we first prove one of our main technical results: for monotone preferences, the IS dynamics always converge.
Then, we exhibit an example of non convergence in which players have individually rational preferences and we show that, always under individually rational preferences, the IS dynamics always converge when starting from a partition with at most three coalitions. 

Some results of this section will rely on the following lemma.
\begin{lemma} \label{lem:size} 
The number of coalitions does not increase during the IS dynamics when $(N,L)$
is a path and players have individually rational preferences.
\end{lemma}

\subsection{Monotone preferences}\label{sec:path:monotone}
Let us note that individual stability and Nash stability are equivalent when the instance is monotone, since no player refuses that another enters her coalition.

We shall prove that, under monotone preferences, the IS dynamics always converge on paths.
Roughly speaking, the proof assumes that a cycle exists in the dynamics and exhibits a contradiction as follows. Consider any player $i$, being the rightmost of coalition $C$, joining in the dynamics the coalition $C'$ on her right. We prove that, in order for player $i$ to go back to her original coalition $C$, the rightmost player in $C'$, say player $i'$, has to join coalition $C''$, being on the right of $i'$. By iteratively repeating this argument, we obtain that the rightmost player of the path should also join the coalition on her right: a contradiction, given that no coalition exists on her right.
It is therefore possible to prove the following theorem.

\def\s{{s}}
\def\t{{t}}
\begin{theorem}\label{thm:paths:monotone} 
Under monotone preferences, the IS dynamics always converge on paths.
\end{theorem}
\begin{proof}[Proof sketch] 
We first need some additional notation.
For $s,t \in \bbN$ with $s \leq t$ we write $[s,t]=\{s,s+1,\ldots,t\}$; we write $[s]=\{1,2,\ldots,s\}$. 
The input graph is a path $P$ where the players are named  
$1, \ldots, n$ 
from left to right. A feasible coalition 
in $P$ is a sequence of consecutive players denoted by $[\ell,r]_P$, where $\ell$ and $r$ are the leftmost and 
and 
rightmost players, respectively. 


Given a cycle  $D=\langle \pi_0,\ldots,\pi_{\alpha} \rangle$ in the IS dynamics, let $|D|=\alpha+1$ denote its length, i.e., the number of deviations. 
Therefore, $D$ is composed of IS deviations $m_0,\ldots,m_{|D|-1}$, where, for every $i=0,\ldots,|D|-1$, deviation $m_i$ at time $t_i$ is from state $\pi_{i}$ to state $\pi_{(i+1) \mod |D|}$.

Given a cycle $D$ in the IS dynamics, and any integers $a, b \in \{0,\ldots,|D|-1\}$, we denote by $[a, b]_D$ the set of integers defined as follows: first of all, let $a'=a \mod |D|$ and $b'=b \mod |D|$; if $b' \geq a'$, then $[a, b]_D  = \{a',\ldots,b'\}$, otherwise, i.e., if $b'<a'$, $[a, b]_D  = \{a',\ldots,|D|-1\} \cup \{0,\ldots,b'\}$. Finally, we say that integer $a$ is \emph{closer (resp., farther) with respect to time $t$} than integer $b$ if $(a-t) \mod |D| < (b-t) \mod |D|$ (resp., if $(a-t) \mod |D| > (b-t) \mod |D|$).\footnote{We are assuming that $r = x \mod |D|$ is defined according to the floored division, i.e., when the quotient is defined as $q=\left\lfloor \frac {x}{|D|} \right\rfloor$, and thus the remainder $r= x - q |D|$ is always non-negative even if $x$ is negative.}

Assume by contradiction that there exists a cycle $D=\langle \pi_0,\ldots,\pi_{|D|-1} \rangle$ in the IS dynamics. 
First of all, notice that, since the number of coalitions can never increase after a deviation of $D$ (by Lemma \ref{lem:size}), 
the number of coalitions of all states in $D$ has to be the same: let $k$ be this number. 
We therefore obtain that, for every $i=0,\ldots,|D|-1$, state $\pi_i$ is composed of $k$ coalitions $C_1^i,\ldots,C_k^i$. 
In the following, we always assume that superscripts 
of $C$ are modulo $|D|$.
Moreover, for any $j=1,\ldots,k$ and $i=0,\ldots,|D|-1$, let $\ell_j^i$ and $r_j^i$ denote the leftmost and rightmost player in 
$C_j^i$, 
respectively, i.e., 
$C_j^i = [\ell_j^i,r_j^i]_P$.

In this sketch, for ease of exposition, we assume that the leftmost player in $P$ making a deviation in $D$ belongs to the first coalition.

For any $j=2,\ldots,k-1$, let $a_{j}$ be the rightmost player in coalition $C_{j}^{\s_{j-1}+1}$ 
(i.e., $a_{j}$ is the rightmost player of the coalition in which player $a_{j-1}$ arrives at time $\s_{j-1}+1$) and let $\s_{j}$ be the closest time with respect to time $\s_{j-1}$ in which $a_j$ moves from coalition 
$C_j^{\s_j}$ to coalition $C_{j+1}^{\s_{j}+1}$. 
Moreover, let $\t_j$ be the farthest time with respect to $\s_j$ in which player $a_j$ comes back to a coalition of index $j$, i.e., she moves from coalition 
$C_{j+1}^{\t_{j}}$ to coalition $C_{j}^{\t_{j}+1}$.

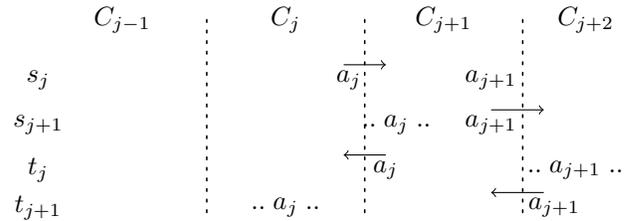
\begin{figure}[ht!]
\begin{center}
\begin{tikzpicture}[xscale=1.4,yscale=1]
\definecolor[named]{drawColor}{gray}{0}

\path[
      draw=drawColor,
      line width= 0.6pt,
      dash pattern=on 1pt off 3pt,
      line cap=round,
    ] (1.6,2.6)--(1.6,0);
\path[
      draw=drawColor,
      line width= 0.6pt,
      dash pattern=on 1pt off 3pt,
      line cap=round,
    ]  (3.1,2.6)--(3.1,0);
\path[
      draw=drawColor,
      line width= 0.6pt,
      dash pattern=on 1pt off 3pt,
      line cap=round,
    ]  (4.6,2.6)--(4.6,0);

\node () at (0.8,2.6) {$C_{j-1}$};
\node () at (2.35,2.6) {$C_{j}$};
\node () at (3.85,2.6) {$C_{j+1}$};
\node () at (5.2,2.6) {$C_{j+2}$};

\node () at (0,1.8) {$\s_{j}$};
\node () at (0,1.2) {$\s_{j+1}$};
\node () at (0,0.6) {$\t_{j}$};
\node () at (0,0.1) {$\t_{j+1}$};

\node () at (2.95,1.8) {$a_{j}$};
\node () at (4.3,1.8) {$a_{j+1}$};

\node () at (3.4,1.2) {$..\;a_{j}\;..$};
\node () at (4.3,1.2) {$a_{j+1}$};

\node () at (3.3,0.6) {$a_{j}$};
\node () at (5.1,0.6) {$..\;a_{j+1}\;..$};

\node () at (2.35,0.1) {$..\;a_{j}\;..$};
\node () at (4.9,0.1) {$a_{j+1}$};

\draw[->] (2.9,2.0)--(3.3,2.0);
\draw[->] (4.3,1.4)--(4.8,1.4);
\draw[<-] (2.9,0.8)--(3.3,0.8);
\draw[<-] (4.3,0.3)--(4.8,0.3);
\end{tikzpicture}

\end{center}
\caption{Claim (i) of the proof by induction.} \label{fig:claim.i}
\end{figure}

\begin{figure}[ht!]
\begin{center}
\begin{tikzpicture}[scale=1.2]
\definecolor[named]{drawColor}{gray}{0}

\path[
      draw=drawColor,
      line width= 0.6pt,
      dash pattern=on 1pt off 3pt,
      line cap=round,
    ]  (1.6,1.3)--(1.6,0);
\path[
      draw=drawColor,
      line width= 0.6pt,
      dash pattern=on 1pt off 3pt,
      line cap=round,
    ]  (2.95,1.3)--(2.95,0);
\path[
      draw=drawColor,
      line width= 0.6pt,
      dash pattern=on 1pt off 3pt,
      line cap=round,
    ]  (4.6,1.3)--(4.6,0);

\node () at (2.35,1.3) {$C_{k-2}$};
\node () at (3.85,1.3) {$C_{k-1}$};
\node () at (5.2,1.3) {$C_{k}$};

\node () at (0.5,0.5) {$\s_{j}$};
\node () at (4.45,0.5) {$a_{j}$};
\node () at (5.5,0.5) {$n$};
\draw[->] (4.35,0.7)--(4.8,0.7);

\end{tikzpicture}
\end{center}
\caption{Claim (ii) of the proof by induction.} \label{fig:claim.ii}
\end{figure}

It is possible to prove the following claims by induction on $j=1,\ldots,k-1$:
\begin{itemize}
\item (i) if $j \leq k-2$, then it holds that (see Figure \ref{fig:claim.i}):
\begin{itemize}
\item (i.a) player $a_{j+1}$ is in coalition 
$C_{x}^{\t_j}$ 
with $x \geq j+2$; 
\item (i.b) $\s_{j+1} \in [\s_{j},\t_{j}]_D$ and  $\t_{j+1} \in [\t_{j},\s_{j}]_D$; 
\item (i.c) player $a_{j}$ is both in coalition 
$C_{j+1}^{\s_{j+1}}$ and in coalition $C_{x}^{\t_{j+1}}$ 
with $x \leq j$; 
\end{itemize}
\item (ii) if $j=k-1$, then player $a_{k-1}$ moves to a coalition of index $k$ (by claim (i.a) holding for $j=k-2$), but she cannot go back to a coalition of index $k-1$ (see Figure \ref{fig:claim.ii}), essentially because player $n$ cannot abandon the last coalition: a contradiction to the fact that $D$ is a cycle in the IS dynamics.~\qedhere
\end{itemize}
\end{proof}

\subsection{Individually rational preferences}\label{sec:IR:path}
In this section, we consider the case when the graph is a path and agents have individually rational preferences. 
We first present an example for which the IS dynamics may not converge when 
the  
initial partition contains four coalitions. 

\begin{example}\label{counterexample:paths}
There are $8$ players, $a,b,c,d,e,f,g,h$ aligned on a path in this order. The preferences are given as follows. 
\begin{itemize}
	\item a: $\{a,b\} \succ \{a\}$
	\item b: $\{b,c,d,e\} \succ \{a,b \}\succ \{b,c\} \succ \{b\}$
	\item c: $\{b,c,d,e\} \succ \{c,d,e\} \succ \{c,d\} \succ \{c\}$ 
	\item d: $\{b,c,d,e\} \succ \{c,d,e\} \succ \{c,d\} \succ \{d,e,f,g\} \succ \{d,e,f\} \succ \{b,c,d\} \succ \{d\}$
	\item e: $\{d,e,f,g\} \succ \{d,e,f\} \succ \{e,f\} \succ \{b,c,d,e\} \succ \{c,d,e\} \succ \{e,f,g\} \succ \{e\}$
	\item f: $\{d,e,f,g\} \succ \{ d,e,f\} \succ \{e,f\} \succ \{f\}$
	\item g: $\{d,e,f,g\} \succ \{g, h \} \succ \{f,g\} \succ \{g\}$
	\item h: $\{g,h\} \succ \{h\}$
\end{itemize}
Note that all coalitions which do not appear are ranked at the same positions as the singleton coalitions. Then, we obtain a cycle as depicted below (deviating players are mentioned above the arrows).
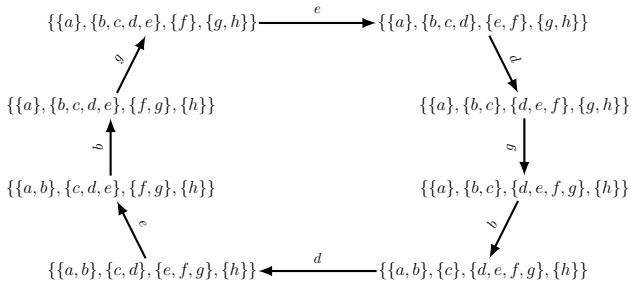
\begin{center}	
\begin{tikzpicture}[scale=0.55, transform shape,every node/.style={minimum size=6mm, inner sep=1pt}]

\node(1) at (-4,3) {\Large $\{\{a\},\{b,c,d,e\},\{f\},\{g,h\}\}$};
\node(2) at (4,3) {\Large $\{\{a\},\{b,c,d\},\{e,f\},\{g,h\}\}$};
\node(3) at (5,1) {\Large $\{\{a\},\{b,c\},\{d,e,f\},\{g,h\}\}$};
\node(4) at (5,-1) {\Large $\{\{a\},\{b,c\},\{d,e,f,g\},\{h\}\}$};
\node(5) at (4,-3) {\Large $\{\{a,b\},\{c\},\{d,e,f,g\},\{h\}\}$};
\node(6) at (-4,-3) {\Large $\{\{a,b\},\{c,d\},\{e,f,g\},\{h\}\}$};
\node(7) at (-5,-1) {\Large $\{\{a,b\},\{c,d,e\},\{f,g\},\{h\}\}$};
\node(8) at (-5,1) {\Large $\{\{a\},\{b,c,d,e\},\{f,g\},\{h\}\}$};

\draw[->, >=latex,thick] (1)--(2) node [midway, above, sloped] (TextNode) {$e$};
\draw[->, >=latex,thick] (2)--(3) node [midway, above, sloped] (TextNode) {$d$};
\draw[->, >=latex,thick] (3)--(4) node [midway, above, sloped] (TextNode) {$g$};
\draw[->, >=latex,thick] (4)--(5) node [midway, above, sloped] (TextNode) {$b$};
\draw[->, >=latex,thick] (5)--(6) node [midway, above, sloped] (TextNode) {$d$};
\draw[->, >=latex,thick] (6)--(7) node [midway, above, sloped] (TextNode) {$e$};
\draw[->, >=latex,thick] (7)--(8) node [midway, above, sloped] (TextNode) {$b$};
\draw[->, >=latex,thick] (8)--(1) node [midway, above, sloped] (TextNode) {$g$};

\end{tikzpicture}
\end{center}
The example can be easily 
extended to 
a larger number of 
coalitions.
\end{example}

Nevertheless, with an initial partition consisting of at most three coalitions and players having individually rational preferences, we are able to prove the following theorem, establishing the convergence of the IS dynamics.  

\begin{theorem}\label{thm:three:paths}
Suppose that the graph $(N,L)$ is a path $P$ and $|\pi_0| \leq 3$. Under individually rational preferences, the IS dynamics always converge.
\end{theorem}

Note that for preferences that are not necessarily monotone, we cannot rely on coalition size alone to reason about player preferences as we did in the proof of Theorem~\ref{thm:paths:monotone}. However, under the assumptions of the above theorem, deviations occur only between the central coalition and the left (or right) coalition, where the left-most (or right-most) player of the coalition is fixed. We  demonstrate that, if we assume the existence of a cycle in the IS dynamics, then players return to the central coalition only if the size of the central coalition decreases, thus yielding a contradiction.

The assumption that the preferences are IR 
is crucial for Theorem~\ref{thm:three:paths}. Without having IR 
preferences, the IS dynamics on paths may not converge even when the initial state 
consists of two coalitions (cf. Example \ref{ex:2coalitions} in 
Appendix). 


\section{Stars} \label{sec:stars}

This section deals with the special case where $(N,L)$ is a star. 
The following lemma will be exploited for proving the convergence of IS dynamics under different settings.

\begin{lemma}\label{lemma:stars}
Suppose that $(N,L)$ is a star and preferences are general. Before converging, the number of deviations of the IS dynamics in which deviating players 
cannot choose to go alone 
is $\mathcal{O}(n^2)$.
\end{lemma}

As an immediate consequence of Lemma \ref{lemma:stars}, we obtain that the IS dynamics always converge under IR 
preferences.

\begin{theorem}\label{IS:stars}
Suppose that $(N,L)$ is a star. 
Under individually rational preferences, the IS dynamics converge in a number of steps which is $\mathcal{O}(n^2)$.
\end{theorem}

In the following, we show that it is possible to obtain a similar result of convergence also for general preferences, by imposing that the dynamics start from a state verifying a suitable property.
To this respect, we will say that a coalition $S \subseteq N$ is {\em individually rational} (IR) if every player $i \in S$ weakly prefers $S$ to $\{i\}$. Moreover, a partition $\pi$ of $N$ is said to be IR if every $S \in \pi$ is individually rational. 

\begin{theorem}\label{IS:stars:IRstate}
Suppose that $(N,L)$ is a star and the initial state is individually rational. 
Under general preferences, the IS dynamics converge in a number of steps which is $\mathcal{O}(n^2)$.
\end{theorem}

It is worth noting 
that starting from an IR state is a weaker requirement than imposing IR preferences; in fact, when players have IR preferences, any state (including the initial one) is IR.
Note that without the assumption that the initial state is IR (and therefore also without the assumption of IR preferences), the IS dynamics may not converge on a star as illustrated by the following example. 
\begin{example}\label{counterexample:star}
Consider a star with center $d$ and leaves $a,b,c$. 
Player $d$ is indifferent between all coalitions and the preferences of the other players are given as follows: 
\begin{itemize}
\item $a: \{a,b,d\} \succ \{a\} \succ \{a,c,d\} \succ \{a,d\} \succ \{a,b,c,d\}$ 
\item $b: \{b,c,d\} \succ \{b\} \succ \{a,b,d\} \succ \{b,d\}\succ \{a,b,c,d\}$
\item $c: \{a,c,d\} \succ \{c\} \succ \{b,c,d\} \succ \{c,d\}\succ \{a,b,c,d\}$
\end{itemize}

A cyclic sequence of states can be as follows: $\{\{a\},\{b,c,d\}\}$ $\substack{c\\ \longrightarrow}$ 
$\{\{a\},\{b,d\},\{c\}\}$ 
$\substack{a\\ \longrightarrow}$ 
$\{\{a,b,d\},\{c\}\}$ 
$\substack{b\\ \longrightarrow}$ 
$\{\{a,d\},\{c\},\{b\}\}$ 
$\substack{c\\ \longrightarrow}$ 
$\{\{a,c,d\},\{b\}\}$ 
$\substack{a\\ \longrightarrow}$ 
$\{\{c,d\},\{b\},\{a\}\}$ 
$\substack{b\\ \longrightarrow}$ 
$\{\{a\},\{b,c,d\}\}$.

\end{example}

Example \ref{example:LB_star} in the appendix 
shows that the bounds on the time of convergence provided in Theorems \ref{IS:stars} and \ref{IS:stars:IRstate} are asymptotically tight.

Finally, 
even when the players' preferences are IR, 
the IS dynamics 
may cycle in a 
star 
with an extra 
node connected to a leaf (cf. Example \ref{ex:almoststar} in the appendix). 

\section{Trees} \label{sec:trees}


In this section, we assume that $(N,L)$ is a tree. We first exhibit an example of non convergence where the valuations 
are non negative and additively separable, which falls in the case of  monotone preferences. Then, we move on to {\em local} additively separable preferences and show that the IS dynamics always converge.
At the end of the section, some additional results holding for the special case of $(N,L)$ being a path, 
under LAS preferences, are provided.

\begin{example} \label{counterexample:mon:tree} Consider the following tree. 
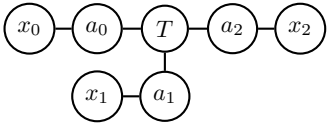
\begin{center}
\begin{tikzpicture}[node distance={10mm}, thick, main/.style = {draw, circle}, scale=0.9, transform shape] 
\node[main] (1) {$x_0$}; 
\node[main] (2) [right of=1]  {$a_0$}; 
\node[main] (3) [right of=2] {$T$}; 
\node[main] (4) [right of=3] {$a_2$}; 
\node[main] (5) [right of=4] {$x_2$}; 
\node[main] (6) [below of=3] {$a_1$}; 
\node[main] (7) [left of=6] {$x_1$};

\draw[-] (1) -- (2);
\draw[-] (2) -- (3);
\draw[-] (3) -- (4);
\draw[-] (4) -- (5);
\draw[-] (3) -- (6);
\draw[-] (6) -- (7);

\end{tikzpicture} 
\end{center}
Suppose the players have additively separable preferences with the following values: for any $i=0,1,2$, $a_i$ has value $1$ for $x_i$, value $2$ for $a_{i+1}$, and $0$ otherwise (subscripts are modulo $3$). Players $T$ and $x_i$ have value $0$ for any other player, for any $i=0,1,2$.
Since all the values are non negative, the described preferences are monotone.  
The IS dynamics may cycle as follows. 
\begin{itemize}
\item $\pi_0= \{\{x_0,a_0\} ,\{T,a_1\} ,\{x_1\},\{a_2,x_2\}\}$. 
\item $\pi_1=\{\{x_0\} ,\{T,a_0,a_1\} ,\{x_1\} ,\{a_2,x_2\}\}$. 
\item $\pi_2=\{\{x_0\} ,\{T,a_0\} ,\{a_1,x_1\},\{a_2,x_2\}\}$. 
\item $\pi_3=\{\{x_0\}, \{T,a_0,a_2\}, \{a_1,x_1\},\{x_2\}\}$. 
\item $\pi_4=\{\{x_0,a_0\} ,\{T,a_2\} ,\{a_1,x_1\},\{x_2\}\}$.  
\item $\pi_5= \{\{x_0,a_0\}, \{T,a_1,a_2\} ,\{x_1\}\{x_2\}\}$. 
\end{itemize}
Here, a sequence of IS deviations is made by $a_0,a_1,a_2,a_0,a_1$, and finally by $a_2$, resulting in the initial state $\pi_0$. 

\end{example}

Note that in Example \ref{counterexample:mon:tree}, one can duplicate the $a_i$ players and reproduce a cycle in the dynamics where all the values are either 1 or 0 (cf. Appendix).

\subsection{Local additively separable preferences}

LAS preferences are considered in this section. Under LAS preferences the IS dynamics can cycle if $(N,L)$ includes a cycle (cf. Example \ref{counterexample:cycles}), but  
the IS dynamics always converge when $(N,L)$ is acyclic.

\begin{theorem} \label{thm:trees}
In a graph hedonic game with LAS preferences over a tree the IS dynamics always converge. 
\end{theorem}
\begin{proof}[Proof sketch] 
 We denote by $N(i)$ the set of neighbours of $i$  (we assume $i \notin N(i)$) and by $L(i)$ the set of edges incident to $i$.
For every feasible partition $\pi$ we define $L_{\pi} = \{(i,j) \in L : \pi(i) = \pi(j)\}$ and $\overline{L}_{\pi} = L\setminus L_{\pi}$ and we refer to them as the set of \emph{built} and \emph{broken} edges in $\pi$, respectively.
For every player $i$ and feasible partition $\pi$ we define $L_{\pi}(i) = \{(i,j): (i,j) \in L(i) \cap  L_{\pi}\}$  and $\overline{L}_{\pi}(i) = L(i)\setminus L_{\pi}(i)$, that are the set of built and broken edges incident to $i$, and we denote by $u_i(\pi)$ the utility of $i$ in $\pi$.

The fact that $(N, L)$ is a tree directly implies that every feasible partition $\pi$ satisfies the  following  property.  
\begin{claim}\label{tree:claim1}
Let $\pi$ be any feasible partition.
Given  any $i\in N$ then, for every $j, j'  \in N(i)$ such that  $j'\neq j$ and $(i,j) \in \overline{L}_{\pi}(i)$, it holds that $\pi(j) \neq \pi(j')$. 
\end{claim}

The deviation of a player $i$ in a feasible partition $\pi$ can be interpreted as selecting a neighboring  vertex $j$ that is incident to an edge in $\overline{L}_{\pi}(i)$.  
This observation enables us to describe the IS dynamics  
in a tree as outlined in Algorithm \ref{alg:is:tree}. 
In order to prove the theorem, we incorporate a labeling scheme $\ell$, where $\ell: L \rightarrow N \cup \{\bot\}$ is a function labeling the edges of the graph as specified in Algorithm \ref{alg:is:tree}.

 \begin{algorithm}                      
\caption{IS dynamics for trees with labeling}         
\label{alg:is:tree}                          
\begin{algorithmic}[1]                  
\REQUIRE a graph hedonic game 
 $(N,L,(\succeq_i)_{i\in N})$, 
such that $(N, L)$ is a tree, and an initial feasible partition $\pi^0$
\ENSURE A feasible partition $\pi$ of $N$ 
\STATE $t \gets 0$
\FOR{each $e \in L$}
\STATE ${\ell^{t}}(e) \gets \bot$ \label{alg:tree:label-init}
\ENDFOR
\WHILE{there exists $(\alpha^t, \beta^t) \in \overline{L}_{\pi^{t}}(\alpha^t)$ 
such that $\alpha^t$ has an IS deviation from $\pi^{t}(\alpha^t)$ to $\pi^{t}(\beta^{t})$
}\label{alg:tree:phase}
\STATE  
$\pi^{t+1} \!\gets\! (\pi^{t} \setminus \{\pi^{t}(\alpha^t),\pi^{t}(\beta^t)\})\cup \{\pi^{t}(\beta^t)\cup \{\alpha^t\}\} \cup \{\, S \mid S~\mbox{is a maximal connected subset of}~\pi^{t}(\alpha^t)\!\setminus\!\{\alpha^t\}\,\}.$\label{alg:tree:ISdeviation}
\FOR{each $e \in L$} 
\IF{ $e \in L_{\pi^t}(\alpha^t)  
\cup \{(\alpha^t,\beta^t)\}$}
\STATE ${\ell^{t+1}}(e) \gets \alpha^t$ \label{alg:tree:label-update1}
\ELSE
\STATE ${\ell^{t+1}}(e) \gets {\ell^{t}}(e)$ \label{alg:tree:label-update2}
\ENDIF
\ENDFOR
\STATE $t \gets t+1$
\ENDWHILE
\RETURN $\pi$
\end{algorithmic}
\end{algorithm}
\noindent Algorithm \ref{alg:is:tree} works as follows. 
At time step $t\geq 0$ player $\alpha^t$ performs a deviation from $\pi^{t}(\alpha^t)$ to $\pi^{t}(\beta^t)$ (see Figure \ref{fig:tree:t-step} in the Appendix); $\pi^{t+1}$ denotes the partition obtained after  this deviation (line \ref{alg:tree:ISdeviation}).
As a consequence of Claim \ref{tree:claim1}, 
$(\alpha^t, \beta^t)$ is the only edge in $L_{\pi^{t+1}}(\alpha^t)$;   
therefore we say that $\alpha^t$  \emph{builds} $(\alpha^t,\beta^t)$ and \emph{breaks} all the edges in $L_{\pi^{t}}(\alpha^t)$ at time $t$.
Moreover, as a consequence of Claim \ref{tree:claim1} and LAS preferences, $u_{\alpha^t}(\pi^{t+1}) =  v_{\alpha^t}(\beta^{t})$.

Algorithm \ref{alg:is:tree} incorporates also a labeling scheme. 
At every time step $t\geq 0$, it assigns a label $\ell^t(e)$ to each edge $e = (i,j)  \in L$; this label gets  value in  $\{\bot, i, j\}$.
The label tracks which of the two endpoints of the edge made the most recent deviation in previous time steps, thereby determining the current status (built or broken) of the edge.
Namely, at time $t=0$  the label of each  edge  $e\in L$ is set  to $\bot$  (line \ref{alg:tree:label-init}), which means that none of the  endpoints of $e$  has   performed any deviation. 
At time $t+1\geq 1$ the labels of all the edges in $L_{\pi^t}(\alpha^t) \cup \{(\alpha^t,\beta^t)\}$ are set to $\alpha^t$ (line \ref{alg:tree:label-update1}), while the labels of all the remaining edges remain unchanged (line \ref{alg:tree:label-update2}). 
The labeling scheme satisfies the following property.

  \begin{claim}\label{tree:claim:one-edge}
  For every $i\in N$ and $t\geq 0$, there exists at most one edge $e$ in $L_{\pi^t}(i)$ such that ${\ell^t(e)} = i$.
  \end{claim}

Now, let us examine how the utility of any player $i$ changes throughout the dynamics given by an execution of Algorithm \ref{alg:is:tree}.
Note that the utility of $i$ is influenced only by the deviations of $i$ and the deviations of the players in $N(i)$. 
Moreover, if  $v_i(j) = 0$, $i$ would never  perform a deviation in which she builds an edge with  $j$. 
Hence, the utility of $i$ strictly increases after each deviation by $i$, 
it may increase after a neighbour $j$ builds an edge with $i$ (it does not change if $v_i(j) = 0$),
it may decrease after a neighbor $j$ breaks an edge with $i$ (it does not change if $v_j(i) = 0$), while 
it strictly decreases after a neighbor $j$ breaks the edge with $i$ when its label is equal to $i$ (in fact, since the  edge has been previously built by $i$, it must hold that $v_i(j) > 0$). 

Next, we will bound the number of times a player $i$ can deviate by building an edge with a given neighbor $j$.
In order to do so, we need to introduce some notation. 
Let  $T_{Q}(i)$ denote the set of time  steps in which $i$ performs a deviation by building an edge connecting herself with any node in $Q\subseteq N(i)$. 
For $j\in N(i)$, we simplify the notation by writing $T_j(i)$ instead of $T_{\{j\}}(i)$.
Trivially, we have  $|T_{Q}(i)| = \sum_{j\in Q}|T_j(i)|$.
We denote by $T(i)$ the set of time steps in which $i$ performs a deviation, i.e., $T(i)  = T_{N(i)}(i)$. 
Since $i$ may also break edges with some neighbors during each deviation, for $j' \in N(i)$ and $Q
\subseteq N(i) \setminus \{j'\}$, we denote by $T^{j'}_{Q}(i)$ the set of time steps in which $i$ performs a deviation  by building an edge with any node in $Q$ while breaking the edge with $j'$ when its label is equal to $j'$.
For $j\in N(i)$, we simplify the notation by writing $T^{j'}_j(i)$ instead of $T^{j'}_{\{j\}}(i)$.
Trivially, we have  $|T^{j'}_{Q}(i)| = \sum_{j\in Q}|T^{j'}_j(i)|$.
We define $T^{j'}(i) = T_{N(i)\setminus \{j'\}}^{j'}(i)$, that is the set of time steps in which $i$ breaks the edge with $j' \in N(i)$ when its label is equal to $j'$. 
Also in this case we have $|T^{j'}(i)| = \sum_{j\in N(i)\setminus \{j'\}} |T^{j'}_{j}(i)|$.
 We are ready to show a bound on the number of times $i$ builds an edge. 


\begin{lemma}\label{tree:lemma:join}
For every $i\in N$ and $Q\subseteq N(i)$, 
\begin{align}
    |T_{Q}(i)|
    &\leq |Q|\left(1 + \sum_{q \in N(i)} |T^{i}(q)|\right). \nonumber
\end{align}  
\end{lemma}

Now, let us consider a player $r$ and let $\hat G$ be the tree $(N,L)$ rooted in $r$.
Our goal is to bound the  number of deviations of $r$. 
We denote by $C_i$ the  set of children of $i$, by $D_i$  the  set of players in the  subtree of  $\hat G$ rooted in $i$ (including $i$) and by $p(i)$ the  parent of $i\neq r$.
Notice that  
$C_r = N(r)$, $D_r = N$, $C_i = N(i) \setminus \{p(i)\}$ for every $i\neq r$, and $C_i = \emptyset$ and $D_i = \{i\}$ for every leaf $i$. 
For every $i$, let $d_i$ be the maximum distance between $i$ and any leaf in $D_i$ ($d_i = 0$ for every leaf $i$). 
For every path $\sigma_{j}^i=\langle j=q_0, q_1, \ldots, q_{k}=i\rangle$ such that $k\geq 0$ and   $q_{h+1}  = p({q_h})$ for every $h=0\ldots,k-1$, we define $m_j^i = \prod_{h=0}^k|C_{q_h}|$ (notice that $m^i_j=0$ for every leaf $j$).

\begin{lemma}\label{tree:moves:i}
For every $i\in N\!\setminus\! \{r\}$,
$
	|T^{p(i)}_{C_i}(i)| \leq   \sum_{j\in D_i} m_j^i
 $.
\end{lemma}

\begin{lemma}\label{tree:moves:r}
For the root $r$, 
$
 |T{(r)}| \leq \sum_{j\in N} m^r_j
 $.
\end{lemma}

\begin{proof}
   	Since, $N(r) = C_r$, from Lemma \ref{tree:lemma:join} we have 
	\begin{align}
   |T(r)| = |T_{C_r}(r)|  
   &\leq |C_r|\big(1 + \sum_{q \in C_r} |T^{r}(q)|\big). \nonumber
    \end{align}
    
Moreover, since $T^{r}(q) = T^{p(q)}_{C_q}(q)$, we can apply Lemma \ref{tree:moves:i} to the previous inequality and obtain 
	\begin{align}
   |T(r)|  
   &\leq |C_r|\big(1 + \sum_{q \in C_r}  \sum_{j\in D_q} m_j^q\big) =  \sum_{j\in D_r} m^r_j. \nonumber
    \end{align}
\end{proof}

The theorem follows from the arbitrary  choice of $r$. 
\end{proof}

\begin{cor}\label{cor:tree} 
In graph hedonic games with LAS preferences, the IS dynamics over a path converge within $2n^2$ deviations.
\end{cor}

The following proposition shows that the IS dynamics can require an exponential number of deviations before converging on trees under LAS preferences, by providing an instance with a player performing a number of deviations matching the bound of Lemma \ref{tree:moves:r}, and thus also showing that this bound is tight. 
Finally, Example \ref{n_square_path} in the appendix 
shows that the IS dynamics can require $\Omega(n^2)$ steps for converging on a path under LAS preferences, thus proving that the bound of Corollary \ref{cor:tree} is asymptotically tight. 


\begin{prop} \label{prop:exp_time}
In a graph hedonic game with LAS preferences over a tree, the IS dynamics may require an exponential number of deviations before converging. 
\end{prop}

\section{Conclusion and future work}

This article draws a picture of convergence issues of the IS dynamics in graph hedonic games. Our results, summarized in Table \ref{table:HGT}, depend on the graph topology and a hierarchy of preferences. Many open questions are left for future works. 

In most of the cases, 
we could upper bound the worst-case number of steps that the IS dynamics may need before converging. 
However, this question remains open for the cases covered by Theorems \ref{thm:paths:monotone} and \ref{thm:three:paths}.

 The IS dynamics rely on \emph{better} responses but note 
that our counterexamples with cyclic sequences hold also with {\em best} responses. 
However, the results showing a lower bound to the number of deviations needed for convergence  (Example \ref{example:LB_star} and Proposition \ref{prop:exp_time}) do not extend to the case of best responses. We in fact conjecture that under best responses the time of convergence can significantly improve.

When a player $i$ leaves her coalition $\pi(i)$ during the IS  dynamics, 
we suppose that the members of $\pi(i) \setminus \{i\}$ reconfigure themselves into a minimum number of coalitions, but a different reconfiguration scheme can be assumed if, for example, $\pi(i) \setminus \{i\}$ is replaced by singletons. An interesting question is whether the reconfiguration 
has a significant impact on the convergence of the IS dynamics. 

This work deals with convergence for any instance and any initial state but if convergence is not always guaranteed, then it is 
worth determining the complexity of deciding it for a given instance, from  
a given initial state.

Finally, a convergence study combining graph topologies and properties of the preferences can be conducted on deviations other than the IS deviations.


\bibliography{IS_refs.bib}

\appendix

\clearpage
\section{Appendix}

\subsection*{Extension of Example \ref{counterexample:cycles}}

For $n \in \bbN$ we write $[n]=\{1,2,\ldots,n\}$. 
The vertex set $N$ is equal to $[n]$,  and $L = \{\{i,i+1\}  : i \in [n-1]\} \cup \{1,n\}$. For $i \in [n-1]$, player $i$ has value $1$ for player $i+1$, and zero otherwise. Player $n$ has value $1$ for player $1$, and zero otherwise. Suppose the initial configuration consists of the three nonempty coalitions: $\{1, \ldots, n-2\}$, $\{n-1\}$ and $\{n\}$. The pattern of the cyclic sequence is such that: (step 1) the (unique) player of the first coalition who is connected with the player of the second coalition (for which she has value 1) decides to join the second coalition,      
(step 2) the unique player of the second coalition who is connected with the player of the third coalition (and for which she has value 1) decides to join the third coalition, and (step 3) the unique player of the third coalition who is connected with the player of the first coalition (and for which she has value 1) decides to join the first coalition. The 3-step pattern can be repeated indefinitely, so the dynamics may never converge. 

\subsection*{Proof of Lemma \ref{lem:size}}
\begin{proof} On a path, the feasible coalitions of a partition $\pi$ are sub-paths. The hypothesis of   
individually rational preferences excludes the possibility that an player leaves her coalition to create a new coalition where she is alone.  Thus, the only possible deviations on a path are that a player on the extreme right (resp., left) of a coalition moves to the nonempty coalition on her right (resp., left). These deviations either maintain the number of coalitions, or decrease the number of coalitions by one unit.      
\end{proof}

\subsection*{Omitted material of Section~\ref{sec:path:monotone}}
We present the full proof of Theorem~\ref{thm:paths:monotone}.
We first need some additional notation.

For $s,t \in \bbN$ with $s \leq t$ we write $[s,t]=\{s,s+1,\ldots,t\}$; we write $[s]=\{1,2,\ldots,s\}$. 
The input graph is a path $P$ where the players are named  
$1, \ldots, n$ 
from left to right. A feasible coalition 
in $P$ is a sequence of consecutive players denoted by $[\ell,r]_P$, where $\ell$ and $r$ are the leftmost and 
and 
rightmost players, respectively. 


Given a cycle  $D=\langle \pi_0,\ldots,\pi_{\alpha} \rangle$ in the IS dynamics, let $|D|=\alpha+1$ denote its length, i.e., the number of deviations. 
Therefore, $D$ is composed of IS deviations $m_0,\ldots,m_{|D|-1}$, where, for every $i=0,\ldots,|D|-1$, deviation $m_i$ at time $t_i$ is from state $\pi_{i}$ to state $\pi_{(i+1) \mod |D|}$.

Given a cycle $D$ in the IS dynamics, and any integers $a, b \in \{0,\ldots,|D|-1\}$, we denote by $[a, b]_D$ the set of integers defined as follows: first of all, let $a'=a \mod |D|$ and $b'=b \mod |D|$; if $b' \geq a'$, then $[a, b]_D  = \{a',\ldots,b'\}$, otherwise, i.e., if $b'<a'$, $[a, b]_D  = \{a',\ldots,|D|-1\} \cup \{0,\ldots,b'\}$. Finally, we say that integer $a$ is \emph{closer (resp., farther) with respect to time $t$} than integer $b$ if $(a-t) \mod |D| < (b-t) \mod |D|$ (resp., if $(a-t) \mod |D| > (b-t) \mod |D|$).\footnote{We are assuming that $r = x \mod |D|$ is defined according to the floored division, i.e., when the quotient is defined as $q=\left\lfloor \frac {x}{|D|} \right\rfloor$, and thus the remainder $r= x - q |D|$ is always non-negative even if $x$ is negative.}

\def\s{{s}}
\def\t{{t}}
\begin{proof}[Proof of Theorem~\ref{thm:paths:monotone}]
Assume by contradiction that there exists a cycle $D=\langle \pi_0,\ldots,\pi_{|D|-1} \rangle$ in the IS dynamics. 
First of all, notice that, since the number of coalitions can never increase after a deviation of $D$ (by Lemma \ref{lem:size}), 
the number of coalitions of all states in $D$ has to be the same: let $k$ be this number. 
We therefore obtain that, for every $i=0,\ldots,|D|-1$, state $\pi_i$ is composed of $k$ coalitions $C_1^i,\ldots,C_k^i$. 
In the following, we always assume that superscripts 
of $C$ are modulo $|D|$.
Moreover, for any $j=1,\ldots,k$ and $i=0,\ldots,|D|-1$, let $\ell_j^i$ and $r_j^i$ denote the leftmost and rightmost player in 
$C_j^i$, 
respectively, i.e., 
$C_j^i = [\ell_j^i,r_j^i]_P$.

Let $a$ be the leftmost player in $P$ making a deviation in $D$, and let $\beta$ be the smallest index of the coalition $a$ belongs to in $D$: there exists a time $\s_\beta \in \{0,\ldots,|D|-1\}$ such that deviation $m_{\s_\beta}$ is performed by player $a_\beta(=a)$ moving from 
coalition 
$C_\beta^{\s_{\beta}}$ to coalition $C_{\beta+1}^{\s_{\beta}+1}$.
Moreover, let $\t_\beta$ be the farthest time with respect to $\s_\beta$ at which player $a_\beta$ comes back to the coalition of index $\beta$, i.e., she moves from coalition 
$C_{\beta+1}^{\t_{\beta}}$ to coalition $C_{\beta}^{\t_{\beta}+1}$. 

For any $j=\beta+1,\ldots,k-1$, let $a_{j}$ be the rightmost player in coalition $C_{j}^{\s_{j-1}+1}$ 
(i.e., $a_{j}$ is the rightmost player of the coalition in which player $a_{j-1}$ arrives at time $\s_{j-1}+1$) and let $\s_{j}$ be the closest time with respect to time $\s_{j-1}$ in which $a_j$ moves from coalition 
$C_j^{\s_j}$ to coalition $C_{j+1}^{\s_{j}+1}$ 
($\s_{j}=\infty$ if $a_j$ never moves from a coalition of index $j$ to a  coalition of index ${j+1}$). 
Moreover, let $\t_j$ be the farthest time with respect to $\s_j$ in which player $a_j$ comes back to a coalition of index $j$, i.e., she moves from coalition 
$C_{j+1}^{\t_{j}}$ to coalition $C_{j}^{\t_{j}+1}$ 
($\t_{j}=\infty$ if $\s_{j}=\infty$, otherwise notice that player $a_j$ has to come back in order to complete the cycle $D$ in the dynamics).  

We now prove the following claims by induction on $j=\beta,\ldots,k-1$:
\begin{itemize}
\item (i) if $j \leq k-2$, then it holds that:
\begin{itemize}
\item (i.a) player $a_{j+1}$ is in coalition 
$C_{x}^{\t_j}$ 
with $x \geq j+2$; 
\item (i.b) $\s_{j+1} \in [\s_{j},\t_{j}]_D$ and  $\t_{j+1} \in [\t_{j},\s_{j}]_D$; 
\item (i.c) player $a_{j}$ is both in coalition 
$C_{j+1}^{\s_{j+1}}$ and in coalition $C_{x}^{\t_{j+1}}$ 
with $x \leq j$; 
\end{itemize}
\item (ii) if $j=k-1$, then player $a_{k-1}$ moves to a coalition of index $k$ (by claim (i.a) holding for $j=k-2$), but she cannot go back to a coalition of index $k-1$: a contradiction to the fact that $D$ is a cycle in the IS dynamics.
\end{itemize}

\begin{figure}[ht]
\begin{center}
\begin{tikzpicture}
[xscale=1.4,yscale=1]
\definecolor[named]{drawColor}{gray}{0}
 \path[
      draw=drawColor,
      line width= 0.6pt,
      dash pattern=on 1pt off 3pt,
      line cap=round,
    ] (1.6,2.6)--(1.6,0);
 \path[
      draw=drawColor,
      line width= 0.6pt,
      dash pattern=on 1pt off 3pt,
      line cap=round,
    ] (3.1,2.6)--(3.1,0);
 \path[
      draw=drawColor,
      line width= 0.6pt,
      dash pattern=on 1pt off 3pt,
      line cap=round,
    ] (4.6,2.6)--(4.6,0);

\node () at (0.8,2.6) {$\cdots\; C_{\beta-1}$};
\node () at (2.35,2.6) {$C_{\beta}$};
\node () at (3.85,2.6) {$C_{\beta+1}$};
\node () at (5.2,2.6) {$C_{\beta+2}$};

\node () at (0,1.8) {$\s_{\beta}$};
\node () at (1.8,1.8) {$\ell_{\beta}^{\s_\beta}$};
\node () at (2.9,1.8) {$a_{\beta}$};
\node () at (4.3,1.8) {$a_{\beta+1}$};

\node () at (3.4,1.2) {$..\;a_{\beta}\;..$};
\node () at (4.3,1.2) {$a_{\beta+1}$};
\node () at (0,1.2) {$\s_{\beta+1}$};

\node () at (3.3,0.6) {$a_{\beta}$};
\node () at (1.8,0.6) {$\ell_{\beta}^{\s_\beta}$};
\node () at (5.1,0.6) {$..\;a_{\beta+1}\;..$};
\node () at (0,0.6) {$\t_{\beta}$};

\node () at (2.35,0.1) {$..\;a_{\beta}\;..$};
\node () at (4.95,0.1) {$a_{\beta+1}$};
\node () at (0,0.1) {$\t_{\beta+1}$};

\draw[->] (2.9,2.0)--(3.3,2.0);
\draw[->] (4.3,1.4)--(4.8,1.4);
\draw[<-] (2.9,0.8)--(3.3,0.8);
\draw[<-] (4.3,0.3)--(4.8,0.3);
\end{tikzpicture}

\end{center}
\caption{Claim (i) of the induction base. Coalitions of $P$ are displayed from left to right, separated with dotted vertical lines. The evolution of the coalitions over time is shown from top to bottom, with time steps given on the extreme left.} \label{fig:base.i}
\end{figure}
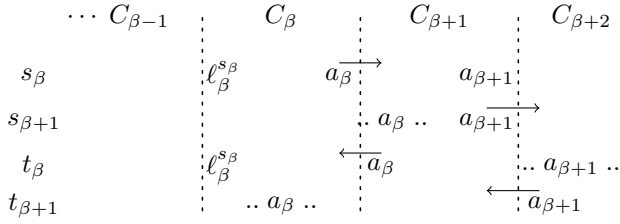

The base of the induction is verified for $j=\beta \leq k-2$ (see Figure \ref{fig:base.i}): player $a_{{\beta}+1}$ has to be in coalition $C_{x}^{\t_{\beta}}$ with $x\geq{\beta}+2$ because the absence of player $a_{{\beta}+1}$ in coalition $C_{{\beta}+1}^{\t_{\beta}}$  is necessary in order for player $a_\beta$ to go back to coalition $C_{{\beta}}^{\t_{\beta}+1}$ (claim i.a). In fact, by deviation $m_{\s_\beta}$, it holds that $[a_\beta, a_{\beta+1}]_P  \succ_{a_\beta} [\ell_{\beta}^{\s_\beta},a_\beta]_P$ and the first $\beta-1$ coalitions\footnote{I.e., the $\beta-1$ coalitions which are on the left of $C_\beta$.} are fixed throughout the entire cycle\footnote{This is due to the fact that $a_\beta$ is the leftmost player who moves.} $D$, implying that $\ell_{\beta}^{\t_{\beta}}=\ell_{\beta}^{\s_\beta}$: by monotonicity, player $a_\beta$ needs that $a_{{\beta}+1}$ is not in coalition $C_{{\beta}+1}^{\t_{\beta}}$ in order to have incentive to go back to coalition $C_{{\beta}}^{\t_{\beta}+1}$. It also follows that $\s_{\beta+1} \in [\s_{\beta},\t_{\beta}]_D$ and  $\t_{\beta+1} \in [\t_{\beta},\s_{\beta}]_D$ (claim i.b), because player $a_{\beta+1}$ is in coalitions $C_{\beta+1}^{\s_\beta} $ and $C_{x}^{\t_\beta}$ with $x \geq \beta+2$, and that player $a_{{\beta}}$ has to be (by the definitions of $\s_{\beta+1}$, $\t_{\beta}$ and $\t_{\beta+1}$) both in coalition $C_{{\beta}+1}^{\s_{\beta+1}}$  and in coalition $C_{{x}}^{\t_{\beta+1}}$ with $x\leq \beta$ (claim i.c). 

\begin{figure}[ht]
\begin{center}
\begin{tikzpicture}[scale=1.2]
\definecolor[named]{drawColor}{gray}{0}

\path[
      draw=drawColor,
      line width= 0.6pt,
      dash pattern=on 1pt off 3pt,
      line cap=round,
    ]  (1.6,3.1)--(1.6,2.1);
\path[
      draw=drawColor,
      line width= 0.6pt,
      dash pattern=on 1pt off 3pt,
      line cap=round,
    ]  (3.1,3.1)--(3.1,2.1);
\path[
      draw=drawColor,
      line width= 0.6pt,
      dash pattern=on 1pt off 3pt,
      line cap=round,
    ]  (4.6,3.1)--(4.6,2.1);

\node () at (1.2,3.1) {$C_1\; \cdots$};
\node () at (2.35,3.1) {$C_{\beta-1}$};
\node () at (3.85,3.1) {$C_{\beta}$};
\node () at (5.2,3.1) {$C_{k}$};

\node () at (0.6,2.5) {$\s_{\beta}$};

\node () at (3.3,2.5) {$\ell_{\beta}^{\s_\beta}$};
\node () at (4.4,2.5) {$a_{\beta}$};
\node () at (5.8,2.5) {$...\;n$};

\draw[->] (4.4,2.7)--(4.8,2.7);
\end{tikzpicture}

\end{center}
\caption{Claim (ii) of the induction base.} \label{fig:base.ii}
\end{figure}
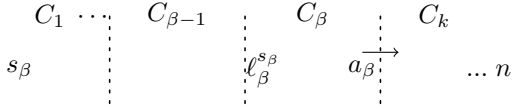

Moreover, if $\beta=k-1$ (see Figure \ref{fig:base.ii}), it holds that player $a_{{\beta}}$ cannot go back to a coalition of index $k-1$, because, by deviation $m_{\s_\beta}$, it holds that $[a_\beta, n]_P  \succ_{a_\beta} [\ell_{\beta}^{\s_\beta},a_\beta]_P$ and  the first $\beta-1$ coalitions are fixed throughout the entire cycle $D$ (claim ii).

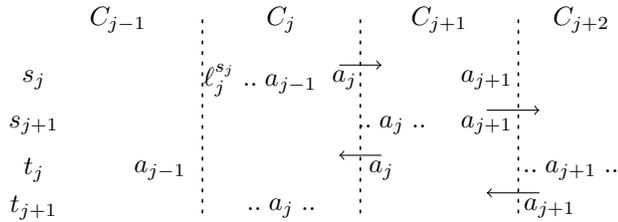
\begin{figure}[ht]
\begin{center}
\begin{tikzpicture}
[xscale=1.4,yscale=1]
\definecolor[named]{drawColor}{gray}{0}

\path[
      draw=drawColor,
      line width= 0.6pt,
      dash pattern=on 1pt off 3pt,
      line cap=round,
    ] (1.6,2.6)--(1.6,0);
\path[
      draw=drawColor,
      line width= 0.6pt,
      dash pattern=on 1pt off 3pt,
      line cap=round,
    ]  (3.1,2.6)--(3.1,0);
\path[
      draw=drawColor,
      line width= 0.6pt,
      dash pattern=on 1pt off 3pt,
      line cap=round,
    ]  (4.6,2.6)--(4.6,0);

\node () at (0.8,2.6) {$C_{j-1}$};
\node () at (2.35,2.6) {$C_{j}$};
\node () at (3.85,2.6) {$C_{j+1}$};
\node () at (5.2,2.6) {$C_{j+2}$};

\node () at (0,1.8) {$\s_{j}$};
\node () at (0,1.2) {$\s_{j+1}$};
\node () at (0,0.6) {$\t_{j}$};
\node () at (0,0.1) {$\t_{j+1}$};

\node () at (2.15,1.8) {$\ell_{j}^{\s_j}\;..\;a_{j-1}$};
\node () at (2.95,1.8) {$a_{j}$};
\node () at (4.3,1.8) {$a_{j+1}$};

\node () at (3.4,1.2) {$..\;a_{j}\;..$};
\node () at (4.3,1.2) {$a_{j+1}$};

\node () at (3.3,0.6) {$a_{j}$};
\node () at (1.2,0.6) {$a_{j-1}$};
\node () at (5.1,0.6) {$..\;a_{j+1}\;..$};

\node () at (2.35,0.1) {$..\;a_{j}\;..$};
\node () at (4.9,0.1) {$a_{j+1}$};

\draw[->] (2.9,2.0)--(3.3,2.0);
\draw[->] (4.3,1.4)--(4.8,1.4);
\draw[<-] (2.9,0.8)--(3.3,0.8);
\draw[<-] (4.3,0.3)--(4.8,0.3);
\end{tikzpicture}

\end{center}
\caption{Claim (i) of the induction step.} \label{fig:step.i}
\end{figure}

As to the induction step, in order to prove the claims for $j=\beta+1,\ldots,k-1$, let us assume the claims hold for $j-1$. 
If $j\leq k-2$ (see Figure \ref{fig:step.i}), then player $a_{j+1}$ has to be in coalition $C_{x}^{\t_{j}}$ with $x \geq j+2$ because the absence of player $a_{j+1}$ in coalition $C_{j+1}^{\t_{j}} $ is necessary in order for player $a_j$ to go back to coalition $C_{j}^{\t_{j}}$ (claim i.a). In fact, by deviation $m_{\s_j}$, it holds that $[a_j, a_{j+1}]_P  \succ_{a_j} [\ell_{j}^{\s_j},a_j]_P$ with $\ell_{j}^{\s_j} \leq a_{j-1}$ as, by the induction hypothesis applied on $j-1$ (claim i.c), we know that player $a_{j-1}$ is in coalition $C_{j}^{\s_j}$: 
by monotonicity, player $a_j$ needs that $a_{{j}+1}$ is not in coalition $C_{{j}+1}^{\t_{j}}$ in order to have incentive to go back to coalition $C_{{j}}^{\t_{j}+1}$ as, again by the induction hypothesis applied on $j-1$ (claim i.c), we know that player $a_{j-1}$ is in coalition $C_{x}^{\t_j}$  with $x \leq j-1$.
It also follows that $\s_{j+1} \in [\s_{j},\t_{j}]_D$ and  $\t_{j+1} \in [\t_{j},\s_{j}]_D$ (claim i.b), because player $a_{j+1}$ is in coalitions $C_{j+1}^{\s_j} $ and $C_{x}^{\t_j}$ with $x\geq j+2$, and that player $a_{{j}}$ has to be (by the definitions of $\s_{j+1}$, $\t_{j}$ and $\t_{j+1}$) both in coalition $C_{{j}+1}^{\s_{j+1}}$  and in coalition $C_{{x}}^{\t_{j+1}}$ with $x\leq j$ (claim i.c).

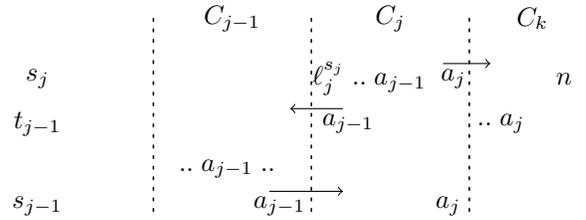
\begin{figure}[ht]
\begin{center}
\begin{tikzpicture}
[xscale=1.4,yscale=1]
\definecolor[named]{drawColor}{gray}{0}

\path[
      draw=drawColor,
      line width= 0.6pt,
      dash pattern=on 1pt off 3pt,
      line cap=round,
    ]  (1.6,2.6)--(1.6,0);
\path[
      draw=drawColor,
      line width= 0.6pt,
      dash pattern=on 1pt off 3pt,
      line cap=round,
    ]  (3.1,2.6)--(3.1,0);
\path[
      draw=drawColor,
      line width= 0.6pt,
      dash pattern=on 1pt off 3pt,
      line cap=round,
    ]  (4.6,2.6)--(4.6,0);

\node () at (2.35,2.6) {$C_{j-1}$};
\node () at (3.85,2.6) {$C_{j}$};
\node () at (5.2,2.6) {$C_{k}$};

\node () at (0.5,1.8) {$\s_{j}$};
\node () at (3.65,1.8) {$\ell_{j}^{\s_j}\;..\;a_{j-1}$};
\node () at (4.45,1.8) {$a_{j}$};
\node () at (5.5,1.8) {$n$};
\draw[->] (4.35,2.0)--(4.8,2.0);

\node () at (3.45,1.2) {$a_{j-1}$};
\node () at (4.9,1.2) {$..\;a_{j}$};
\node () at (0.5,1.2) {$\t_{j-1}$};
\draw[<-] (2.9,1.4)--(3.4,1.4);

\node () at (2.3,0.6) {$..\;a_{j-1}\;..$};

\node () at (2.8,0.1) {$a_{j-1}$};
\node () at (4.4,0.1) {$a_{j}$};
\node () at (0.5,0.1) {$\s_{j-1}$};
\draw[->] (2.7,0.3)--(3.4,0.3);

\end{tikzpicture}
\end{center}
\caption{Claim (ii) of the induction step.} \label{fig:step.ii}
\end{figure}
Moreover, if $j=k-1$ (see Figure \ref{fig:step.ii}), it holds that player $a_{{j}}$ cannot go back to a coalition of index $k-1$. In fact, by deviation $m_{\s_j}$, it holds that $[a_j, n]_P  \succ_{a_j} [\ell_{j}^{\s_j},a_j]_P$ with $\ell_{j}^{\s_j} \leq a_{j-1}$ as, by the induction hypothesis applied on $j-1$ (claim i.c), we know that player $a_{j-1}$ is in coalition $C_{j}^{\s_j}$: by monotonicity, player $a_j$ cannot go back to a coalition of index $j$ as, by the induction hypothesis applied on $j-1$ (claim i.b), we know that, whenever it exists, $\t_j \in [\t_{j-1},\s_{j-1}]_D$ and in all the timestamps in this interval, player $a_{j-1}$ is in coalition of index at most $j-1$ (claim ii).
\end{proof}

\subsection*{Omitted material of Section~\ref{sec:IR:path}}

We present the full proof of Theorem~\ref{thm:three:paths}. 
We first note that when the initial state consists of at most two coalitions, the IS dynamics always converge. 

\begin{lemma}\label{lem:2}
Suppose that $(N,L)$
is a path $P$ and players have individually rational preferences. 
If $|\pi_0| \leq 2$, then the IS dynamics converge.
\end{lemma}
\begin{proof}
Suppose that $|\pi_0| \leq 2$. Without loss of generality, suppose that player $x$ deviates from the right coalition $[x,n]_P$ to the left coalition $[1,x-1]_P$, resulting in a new partition $\pi_1$. This leaves the coalition $[x+1,n]_P$ on the right. 
In this new partition $\pi_1$, the only player who can potentially deviate  is the border player $x+1$ of the right coalition. Since the right coalition keeps getting smaller with each deviation, the IS dynamics must eventually terminate.
\end{proof}

By the above lemma, it suffices to consider the case when the graph is a path and $|\pi_0|= 3$. For each state $\pi_i$ with $|\pi_i|=3$, note that $\pi_i$ partitions the players into the left coalition including the left-most player $1$, the right coalition including the right-most player $n$, and the central coalition including none of these players. In the sequel, we write $C^i$ as the central coalition of $\pi_i$; we denote by $\ell^i$ the leftmost player of the central coalition, i.e., $\ell^i= \min_{j \in C^i} j$.

Before we proceed, we make the following observations. First, players who deviate to the left-most coalition strictly prefer that coalition to the previous coalitions when they become the left-end player; second, if players from the left and the right coalition consecutively deviate to the central coalition, then such dynamics should converge.

\begin{lemma}\label{lem:3}
Suppose that the graph $(N,L)$ is a path $P$. Let $T=[1,x-1]_P$ be the leftmost coalition and $S=[x,y_0]_P$ be the coalition right next to $T$ in some state of the IS dynamics. Suppose that $k$ players $y_1,y_2,\ldots,y_k$ from the coalition on the right of $S$ consecutively join $S$ and then $x$ deviates from $S \cup \{y_1,y_2,\ldots,y_k\}$ to the left coalition $T$. Then, we have $[1,x]_P \succ_{x} [x,y_j]_P$ for any $j=0,1,\ldots,k$. 
\end{lemma}
\begin{proof}
Player $x$ keeps accepting each player $y_j$ for $j=1,2,\ldots,k$; moreover, player $x$ strictly prefers $[1,x]_P$ to the coalition $[x,y_k]_P$. Hence, $[1,x]_P \succ_{x} [x,y_{k}]_P \succeq_{x} [x,y_{k-1}]_P \succeq_{x} \cdots \succeq_{x} [x,y_{0}]_P.$
\end{proof}

\begin{lemma}\label{lem:4}
Suppose that the graph $(N,L)$ is a path $P$, agents have individually rational preferences, and $|\pi_0| =3$. If there exists a pair of players $x$ and $y$ such that $x$ deviates from the left coalition (respectively, the right coalition) to the central coalition and subsequently player $y$ deviates from the right coalition (respectively, the left coalition) to the central coalition, then the IS dynamics converge. 
\end{lemma}
\begin{proof}
Without loss of generality, suppose that player $x$ deviates from the left coalition $[1,x]_P$ to the central coalition $[x+1,y-1]_P$, resulting in a new partition $\pi_{t+1}$; then player $y$ deviates from the right coalition $[y,n]_P$ to the central coalition $[x,y-1]_P$, resulting in a new partition $\pi_{t+2}$. 
In this new partition $\pi_{t+2}$, the only players who can potentially deviate are the border players $x-1$ and $y+1$ on the left and the right coalition. Since the right coalition and the left coalition keep getting smaller with each deviation, the IS dynamics must eventually terminate.
\end{proof}

We prove that the last player who leaves the central coalition comes back to it only if the central coalition gets smaller.

\begin{proof}[Proof of Theorem~\ref{thm:three:paths}]
Assume by contradiction that there exists a cycle $D=\langle \pi_0,\ldots,\pi_{|D|-1} \rangle$ in the IS dynamics. Notice that by Lemmas \ref{lem:size} and \ref{lem:2}, any $\pi_t$ partitions the players into exactly $3$ coalitions. Now, consider the left-most player $\ell^*$ who moves between the left and central coalitions, i.e., $\ell^*=\min_{t'=0,\ldots,|D|-1} \ell^{t'}$. 

We will prove the following claim: Suppose that $\ell^*$ arrives at the left coalition at time $s+1$ and $t$ is the closest time with respect to $s$ at which $\ell^*$ comes back to the central coalition. 
Then, the central coalition at time $t$ is strictly smaller than the central coalition at time $s$ (i.e., $C^t \subsetneq C^s$).

Without loss of generality, assume that $s<t$. Hence, we have $\ell^*=\ell^s=\ell^t$ and $\ell^*<\ell^{t'}$ for every $t' \in [s+1,t-1]$. 
Let $M(s,t)$ be the set of players who move across the central and the left coalitions between step $s$ and $t$, i.e., 
\[
M(s,t)=\{\, j \in N \mid \exists t',t'' \in [s,t]: j = \ell^{t'}  \land j<\ell^{t''}\,\}.
\]
We will prove by induction on $|M(s,t)|$ that the left most player $\ell^*$ always comes back to a smaller central coalition than the one he leaves. 

Consider first the base case when $|M(s,t)|=1$, namely, $\ell^*$ is the only player who move across the central and the left coalitions between step $s$ and $t$. By Lemma \ref{lem:4}, no player moves from the right coalition to the central coalition between steps $s$ and $t$. Thus, $C^t \subseteq C^s$. 
Moreover, since $\ell^*$ strictly prefers the coalition $C^t$ to the coalition $[1,\ell^*]_P$ and the coalition $[1,\ell^*]_P$ to the coalition $C^s$, we have $C^t \succ_{\ell^*} C^s$ and hence $C^t \neq C^s$, which implies $C^t \subsetneq C^s$. 

Now, suppose that our claim holds for $|M(s,t)|=1,2,\ldots,h-1$ and we prove it for the case when $|M(s,t)|=h$. We let $x$ be the right neighbor of $\ell^*$, i.e., $x=\ell^*+1$. 
Since $|M(s,t)| \geq 2$ and $\ell^*$ is the leftmost player who moves from the central coalition to the left coalition, player $x$ moves from the central coalition to the left coalition between steps $s$ and $t$; suppose that $x$ stays in the central coalition between step $s$ and $s_1$ and $x$ arrives at the left coalition at time $s_1+1$, i.e., $\ell^{t'}=x$ for every $t' \in [s+1,s_1-1]$ and $x < \ell^{s_1}$. Further, in order for player $\ell^*$ to come back to the central coalition, player $x$ needs to go back to the central coalition between steps $s_1+1$ and $t$; let $t_1$ be the closest time with respect to $s_1+1$ at which $x$ arrives at the central coalition, i.e., $x<\ell^{t'}$ for every $t' \in [s_1+1,t_1-1]$ and $x = \ell^{t_1}$. We define steps $s_2,t_2,\ldots, s_m,t_m$ in an analogous way, and let $t_m$ be the furthest time at which player $x$ arrives at the central coalition with respect to time $s$, i.e., $x=\ell^{t_m}$ and $x \in C^{t'}$ for every $t' \in [t_m,t]$. 

There are at most $h-1$ players who move between the central and the left coalition between step $s_1$ to $t_1$; thus applying the induction hypothesis to player $x$ implies that
\[
C^{t_1} \subsetneq C^{s_1}. 
\]
In fact, we have
\begin{align}\label{eq0}
C^{t_1} \subsetneq C^s \setminus \{\ell^*\}.
\end{align}
To see this, consider two cases. First, if the size of the central coalition decreases from $s$ to $s_1$, i.e., $C^{s_1} \subseteq C^s \setminus \{\ell^*\}$, we clearly have $C^{t_1} \subsetneq C^{s_1} \subseteq C^s \setminus \{\ell^*\}$. Second, consider the case when the size of the central coalition does not decrease between $s$ and $s_1$, i.e., $C^s \setminus \{\ell^*\} \subseteq C^{s_1}$. This means that some players from the right coalition consecutively join the central coalition and then $x$ deviates from the central coalition to the left coalition; by Lemma \ref{lem:3}, player $x$ strictly prefers $[1,x]_P$ to any central coalition $C^{t'}$ with $t' \in [s,s_1]$. Combining this with the fact that player $x$ strictly prefers $C^{t_1}$ to $[1,x]_P$, we get $C^{t_1} \neq C^{t'}$ for any $t' \in [s,s_1]$. Since $C^{t_1} \subsetneq C^{s_1}$ and $\ell^s= \ell^* < x = \ell^{t_1}$, we obtain $C^{t_1} \subsetneq C^s \setminus \{\ell^*\}$.

Similarly, there are at most $h-1$ players who move between the central and the left coalition from the step $s_j$ to $t_j$; thus applying the induction hypothesis to player $x$ implies that for each $j=2,\ldots,m$, we have
\begin{align}\label{eq1}
C^{t_j} \subsetneq C^{s_j}. 
\end{align}
In addition, for each $j=1,2,\ldots,m-1$, player $x$ comes back to the central coalition at step $t_j$ and leaves it at step $s_{j+1}$; by Lemma \ref{lem:4}, no player can move from the right to the central coalition between these steps; thus,
\begin{align}\label{eq2}
C^{s_{j+1}} \subsetneq C^{t_j}. 
\end{align}
Combining \eqref{eq0}, \eqref{eq1} and \eqref{eq2}, we get 
\[
C^{t_m}\subsetneq C^s \setminus \{\ell^*\}.
\]
Again, no player can move from the right to the central coalition between steps $t_m$ and $t$ due to Lemma \ref{lem:4}; hence, we have $C^t \setminus \{\ell^*\}\subseteq C^{t_m}$, implying that
\[
C^t\subsetneq  C^s.
\]

Finally, we have shown that every time $\ell^*$ comes back to the central coalition and the size of the coalition gets smaller than the one $\ell^*$ leaves; further, no player moves from the right to the central coalition by Lemma \ref{lem:4} until $\ell^*$ leaves the central coalition again. Hence, the size of the central coalition he leaves strictly decreases and thus the dynamics cannot repeat the same configuration, a contradiction.   
\end{proof}


The assumption that the preferences are  individually rational 
is very important for Theorem~\ref{thm:three:paths}. 
Without having IR 
preferences, the following example shows that 
the IS dynamics on paths may not converge even when the initial state 
consists of two coalitions.

\begin{example} \label{ex:2coalitions}
Let us create additional players $\alpha,\alpha'$ placed to the right of the path in Example~\ref{counterexample:paths}. Suppose that player $\alpha$ strictly prefers $\{\alpha,\alpha'\}$ to $\{a,b,c,d,e,f,g,h,\alpha\}$, and player $\alpha'$ is indifferent among all coalitions. Moreover, we modify the preferences of the original players so that the coalitions $\{a,b,c,d,e,f,g,h\}$ and $\{a,b,c,d,e\}$ are not individually rational for player $f$ and $a$, respectively.

Starting with an individually rational partition $\pi=\{\{a,b,c,d,e,f,g,h,\alpha\}, \{\alpha'\}\}$, player $\alpha$ deviates to $\{\alpha'\}$, which gives a partition $\pi$ that consists of $\{a,b,c,d,e,f,g,h\}$ and $\{\alpha,\alpha'\}$. Then $f$ may deviate to a singleton coalition, followed by the deviation of $a$ into a singleton coalition; the resulting partition consists of $\{\alpha,\alpha'\}$ and the partition $\{\{a\},\{b,c,d,e\},\{f\},\{g,h\}\}$ of the remaining players, which is one of the states of Example~\ref{counterexample:paths}. This means that we enter the cyclic sequence.
\end{example}

\subsection*{Omitted material of Section  \ref{sec:stars}}

\subsubsection*{Proof of Lemma  \ref{lemma:stars}}

\begin{proof}
Let $c$ denote the central player of the star. We refer to the coalition to which $c$ belongs as the \emph{central} coalition. 
By the assumption that in the IS dynamics players cannot go alone, we know that $T \neq \emptyset$ (cf. Line~\ref{line:ISdeviation} of Algorithm~\ref{alg:is:general}).
Thus, there are only two possible types of deviations: denoting by $j$ the player performing the deviation, either (i) $j=c$ and $T$ consists of a leaf of the graph $(N,L)$, or (ii) $j$ is a leaf of the graph and $T$ is the central coalition. 
Notice that both types of deviations (i) and (ii) concern the central player who is never worst off. There are at most $n-1$ coalitions $\{i\}$, with  $i \in N \setminus \{c\}$, to which the central player $c$ can deviate. Moreover, there are at most $n-1$ players in $N \setminus \{c\}$ who can join the central coalition, and no such player can leave the central coalition (by the assumption that in the IS dynamics players cannot go alone).
Hence, there are at most $n-1$ deviations of type (i), and at most $n-1$ deviations of type (ii) between two (consecutive) deviations of type (i) or after the last deviation of type (i). Hence, the number deviations in the considered IS dynamics is $\mathcal{O}(n^2)$.
\end{proof}

\subsubsection*{Proof of Theorem  \ref{IS:stars}}

\begin{proof}
The result follows from Lemma \ref{lemma:stars}.  When the preferences are IR, 
no agent $i$ prefers the singleton coalition $\{i\}$ to any other coalition $C \supset \{i\}$, and therefore it cannot happen that an agent chooses to go alone.
\end{proof}

\subsubsection*{Proof of Theorem  \ref{IS:stars:IRstate}}

\begin{proof}
First of all, notice that, by definition, in an IR state no agent can deviate to an empty coalition, i.e., no agent can choose to go alone. 
Therefore, by Lemma \ref{lemma:stars}, if it holds that all the states of the IS dynamics are IR, then the IS dynamics converge in a number of steps polynomial in $n$.
Thus, in order to prove the claim, in the following we show that a deviation starting from a IR state leads to another IR state. 

Let $c$ denote the central player of the star. We refer to the coalition to which $c$ belongs as the \emph{central} coalition.
Indeed, suppose that $\pi$ is IR and there is a feasible deviation of player $j \in N$ to $T \in \pi \cup \{\emptyset\}$. Then, recalling that $T \neq \emptyset$ because $\pi$ is IR, either (i) $j=c$ and $T$ consists of a leaf of the graph $(N,L)$, or (ii) $j$ is a leaf of the graph and $T$ is the central coalition. In either case, it can be easily verified that the new partition $\pi'$ that results from the deviation of $i$ to $T$ (cf. Line~\ref{line:ISdeviation} of Algorithm~\ref{alg:is:general}) is IR.
\end{proof}

\subsubsection*{Example showing tightness of Theorems \ref{IS:stars} and \ref{IS:stars:IRstate}}

\begin{example}\label{example:LB_star}
Consider a star with $2t$ leaf players belonging to sets $X=\{x_1,\ldots,x_t\}$ and $Y=\{y_1,\ldots,y_t\}$, in which the preferences of the central player $c$ are as follows: given two candidate coalitions $C$ and $C'$ for $c$, i.e., $C, C' \in \calF(c)$ with $C \neq C'$, $c$ prefers coalition $C$ to coalition $C'$, i.e., $C \succ_c C'$,  if (at least) one of the following conditions is met:
\begin{itemize}
    \item $C' = \{c\}$;
    \item $|C \cap X|=1$ and $|C' \cap X| \neq 1$;
    \item $C \cap X=\{x_i\}$ , $C' \cap X=\{x_{i'}\}$ and $i>i'$;
    \item $C \cap X=C' \cap X=\{x_i\}$ and $|C \cap Y| > |C' \cap Y|$.  
\end{itemize}
If, given two candidate coalitions $C$ and $C'$ for $c$, the above conditions neither imply $C \succ_c C'$ nor $C' \succ_c C$, we assume that $C \sim_c C'$, i.e., $c$ is indifferent between $C$ and $C'$. Moreover, we assume that any leaf player $l \in X \cup Y$ strictly prefers any coalition in $\calF(l) \setminus \{\{l\}\}$ to the singleton coalition $\{l\}$ and is indifferent among all coalitions in $\calF(l) \setminus \{\{l\}\}$. Notice that the above define preferences are individually rational.
For the sake of simplicity, in the following we identify a partition by its central coalition, i.e., by the coalition of its central player $c$ (notice that all other coalitions have to be singletons). The initial state of the IS dynamics is $\{c\}$. Consider the following sub-sequence of $t+1$ partitions inside the IS dynamics, depending on parameter $i$, for $i=1,\ldots,t$: $\{c,x_i\}, \{c,x_i,y_1\}, \{c,x_i,y_1,y_2\},\ldots, \{c,x_i,y_1,y_2,\ldots,y_t\} $ and notice that it can be obtained starting form $\{c,x_i\}$ and letting players $y_1,\ldots,y_t$ deviate by joining the central coalition.
The whole IS dynamics is obtained by chaining the initial state $\{c\}$ with the sub-sequences with parameters $1,\ldots,t$. 
In fact, the first state $\{c,x_1\}$ of the sub-sequence with parameter $1$  can be obtained by letting player $c$ deviate from the initial state $\{c\}$ to coalition $\{x_1\}$, while, for $i=2,\ldots,t$, the first state $\{c,x_i\}$ of the sub-sequence with parameter $i$  can be obtained by letting player $c$ deviate from the final state $\{c,x_{i-1},y_1,y_2,\ldots,y_t\}$ of the sub-sequence with parameter $i-1$ to coalition $\{x_i\}$. 
Overall, the described IS dynamics is composed by $t(t+1)$ deviations. Since $n=2t+1$, we have  convergence after a number of steps being $\Omega(n^2)$.
\end{example}

\subsubsection*{Non convergence in an almost star}

\begin{example} \label{ex:almoststar}Consider the instance described in Figure \ref{fig:almost_star} where the preferences are: 
\begin{itemize}
\item $a: \{a,b\} \succ \{a,e\} \succ \{a,b,d\} \succ  \{a,d\}  \succ \{a\}$
\item $b: \{a,b,d\} \succ \{b,c\} \succ \{a,b\} \succ  \{b\}$
\item $c: \{b,c\} \succ \{c\}$
\item $d: \{a,b,d\} \succ \{a,d\} \succ \{d\}$
\item $e: \{a,e\} \succ \{e\}$

\end{itemize}
The initial state is $\{\{a,b\},\{c\},\{d\},\{e\} \}$. Player $b$ moves from $\{a,b\}$ to $\{c\}$ giving $\{\{a\},\{b,c\},\{d\},\{e\}\}$. Player $a$ moves from $\{a\}$ to $\{d\}$ giving $\{\{a,d\},\{b,c\},\{e\}\}$. Player $b$ moves from $\{b,c\}$ to $\{a,d\}$ giving $\{\{a,b,d\},\{c\},\{e\}\}$. Player $a$ moves from $\{a,b,d\}$ to $\{e\}$ giving $\{\{a,e\},\{b\},\{c\},\{d\}\}$. Player $a$ moves from $\{a,e\}$ to $\{b\}$ giving $\{\{a,b\},\{c\},\{d\},\{e\}\}$, which is the initial state.
\end{example}

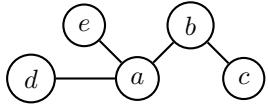
\begin{figure}
\begin{center}
\begin{tikzpicture}[node distance={10mm}, thick, main/.style = {draw, circle}]
\node[main] (1) {$c$}; 
\node[main] (2) [above left of=1]  {$b$}; 
\node[main] (3) [below left of=2] {$a$}; 
\node[main] (4) [above left of=3] {$e$}; 
\node[main] (5) [below left of=4] {$d$};

\draw[-] (1) -- (2);
\draw[-] (2) -- (3);
\draw[-] (4) -- (3);
\draw[-] (3) -- (5);

\end{tikzpicture} 
\end{center}
\caption{An ``almost'' star where the IS dynamics can cycle.} \label{fig:almost_star}
\end{figure}

\subsection*{Omitted material of Section  \ref{sec:trees}}


\subsubsection*{Extension of Example  \ref{counterexample:mon:tree}}
\phantom{.}

\noindent From Example \ref{counterexample:mon:tree}, one can duplicate the $a_i$ players and reproduce a cycle in the IS dynamics where all the values are either 1 or 0.

Concretely, insert a vertex $b_i$ between $a_i$ and $x_i$, for any $i=0,1,2$. Then, both $a_i$ and $b_i$ have utility 1 for $x_i$, $a_{i+1}$, and $b_{i+1}$, and $0$ otherwise, for any $i=0,1,2$ (subscripts are modulo $3$).  
Then, the IS dynamics may cycle as follows (deviations are made, in sequence, by $a_0,b_0,b_1,a_1,a_2,b_2,b_0,a_0,a_1,b_1,b_2,a_2$. 
\begin{itemize}
\item $\pi_0= \{\{x_0,b_0,a_0\} ,\{T,a_1,b_1\} ,\{x_1\},\{a_2,b_2,x_2\}\}$. 

\item $\pi_1= \{\{x_0,b_0\} ,\{T,a_0,a_1,b_1\} ,\{x_1\},\{a_2,b_2,x_2\}\}$. 
\item $\pi_2= \{\{x_0\} ,\{T,a_0,b_0,a_1,b_1\} ,\{x_1\},\{a_2,b_2,x_2\}\}$. 
\item $\pi_3= \{\{x_0\} ,\{T,a_0,b_0,a_1\} ,\{x_1,b_1\},\{a_2,b_2,x_2\}\}$. 
\item $\pi_4= \{\{x_0\} ,\{T,a_0,b_0\} ,\{a_1,x_1,b_1\},\{a_2,b_2,x_2\}\}$. 

\item $\pi_5= \{\{x_0\} ,\{T,a_0,b_0,a_2\} ,\{a_1,x_1,b_1\},\{b_2,x_2\}\}$. 
\item $\pi_6= \{\{x_0\} ,\{T,a_0,b_0,a_2,b_2\} ,\{a_1,x_1,b_1\},\{x_2\}\}$. 
\item $\pi_7= \{\{x_0,b_0\} ,\{T,a_0,a_2,b_2\} ,\{a_1,x_1,b_1\},\{x_2\}\}$. 
\item $\pi_8= \{\{x_0,b_0,a_0\} ,\{T,a_2,b_2\} ,\{a_1,x_1,b_1\},\{x_2\}\}$. 

\item $\pi_9= \{\{x_0,b_0,a_0\} ,\{T,a_1,a_2,b_2\} ,\{x_1,b_1\},\{x_2\}\}$. 
\item $\pi_{10}= \{\{x_0,b_0,a_0\} ,\{T,a_1,b_1,a_2,b_2\} ,\{x_1\},\{x_2\}\}$. 
\item $\pi_{11}= \{\{x_0,b_0,a_0\} ,\{T,a_1,b_1,a_2\} ,\{x_1\},\{b_2,x_2\}\}$. 

\end{itemize}


\bigskip
\subsubsection*{Proof of Theorem \ref{thm:trees}}\phantom{h}

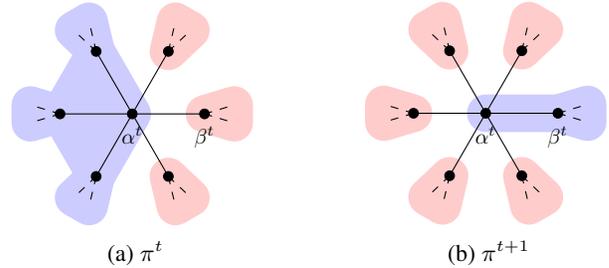
\begin{figure}[h]
    \centering
\begin{subfigure}{0.2\textwidth}
    \centering
    \begin{tikzpicture}[scale=0.8,transform shape]


	\def\ra{1.2} 
	\def\rax{1.7} 
	\def\diff{5}
	\coordinate (X) at (0:0); 
	\coordinate (P1) at (0:\ra);
	\coordinate (P2) at (60:\ra);
	\coordinate (P3) at (120:\ra);
	\coordinate (P4) at (180:\ra);
	\coordinate (P5) at (240:\ra);
	\coordinate (P6) at (300:\ra);
	
	\coordinate (PX1) at (0-\diff:\rax);
	\coordinate (PX2) at (60-\diff:\rax);
	\coordinate (PX3) at (120-\diff:\rax);
	\coordinate (PX4) at (180-\diff:\rax);
	\coordinate (PX5) at (240-\diff:\rax);
	\coordinate (PX6) at (300-\diff:\rax);

	\coordinate (PY1) at (0+\diff:\rax);
	\coordinate (PY2) at (60+\diff:\rax);
	\coordinate (PY3) at (120+\diff:\rax);
	\coordinate (PY4) at (180+\diff:\rax);
	\coordinate (PY5) at (240+\diff:\rax);
	\coordinate (PY6) at (300+\diff:\rax);
	
	\def\dim{0.08} 
	
	\tikzstyle{coalition-blue}=[draw=blue!20,fill=blue!20, line width=0.5cm,line cap=round,line join=round];
	
	\tikzstyle{coalition-red}=[draw=red!20,fill=red!20, line width=0.5cm,line cap=round,line join=round];
	
	\path[coalition-blue]
(X)--(P3)--(P4)--(P5)--cycle;

	\path[coalition-blue]
(P3)--(PY3)--(PX3)--cycle;

	\path[coalition-blue]
(P4)--(PY4)--(PX4)--cycle;

	\path[coalition-blue]
(P5)--(PY5)--(PX5)--cycle;

	\path[coalition-red]
(P6)--(PY6)--(PX6)--cycle;

	\path[coalition-red]
(P1)--(PY1)--(PX1)--cycle;

	\path[coalition-red]
(P2)--(PY2)--(PX2)--cycle;
	
	\draw[fill=black, radius=\dim] 
	(X) circle[radius=\dim]
	node [anchor=north, outer ysep=0.1cm] {$\alpha^t$} 
	(P1) circle[radius=\dim]
	node [anchor=north, outer ysep=0.1cm] {$\beta^t$}
	(P2) circle[radius=\dim]  
	(P3) circle[radius=\dim] 
	(P4) circle[radius=\dim] 
	(P5) circle[radius=\dim] 
	(P6) circle[radius=\dim];
	
	\draw (X) -- (P1);
	\draw (X) -- (P2);
	\draw (X) -- (P3);
	\draw (X) -- (P4);
	\draw (X) -- (P5);
	\draw (X) -- (P6);

	\draw[dashed] (P1) -- (PX1);
	\draw[dashed] (P2) -- (PX2);
	\draw[dashed] (P3) -- (PX3);
	\draw[dashed] (P4) -- (PX4);
 	\draw[dashed] (P5) -- (PX5);
	\draw[dashed] (P6) -- (PX6);
 	\draw[dashed] (P1) -- (PY1);
	\draw[dashed] (P2) -- (PY2);
 	\draw[dashed] (P3) -- (PY3);
	\draw[dashed] (P4) -- (PY4);
	\draw[dashed] (P5) -- (PY5);
	\draw[dashed] (P6) -- (PY6);

\end{tikzpicture}
    \caption{$\pi^t$}
    \label{fig:tree:t-step:A}
\end{subfigure}
\quad
\quad
\quad
\begin{subfigure}{0.2\textwidth}
    \centering
    \begin{tikzpicture}[scale=0.8,transform shape]


	\def\ra{1.2} 
	\def\rax{1.7} 
	\def\diff{5}
	\coordinate (X) at (0:0); 
	\coordinate (P1) at (0:\ra);
	\coordinate (P2) at (60:\ra);
	\coordinate (P3) at (120:\ra);
	\coordinate (P4) at (180:\ra);
	\coordinate (P5) at (240:\ra);
	\coordinate (P6) at (300:\ra);
	
	\coordinate (PX1) at (0-\diff:\rax);
	\coordinate (PX2) at (60-\diff:\rax);
	\coordinate (PX3) at (120-\diff:\rax);
	\coordinate (PX4) at (180-\diff:\rax);
	\coordinate (PX5) at (240-\diff:\rax);
	\coordinate (PX6) at (300-\diff:\rax);

	\coordinate (PY1) at (0+\diff:\rax);
	\coordinate (PY2) at (60+\diff:\rax);
	\coordinate (PY3) at (120+\diff:\rax);
	\coordinate (PY4) at (180+\diff:\rax);
	\coordinate (PY5) at (240+\diff:\rax);
	\coordinate (PY6) at (300+\diff:\rax);
	
	\def\dim{0.08} 
	
	\tikzstyle{coalition-blue}=[draw=blue!20,fill=blue!20, line width=0.5cm,line cap=round,line join=round];
	
	\tikzstyle{coalition-red}=[draw=red!20,fill=red!20, line width=0.5cm,line cap=round,line join=round];
	
	\path[coalition-blue]
(X)--(P1);

	\path[coalition-blue]
(P1)--(PY1)--(PX1)--cycle;

	\path[coalition-red]
(P2)--(PY2)--(PX2)--cycle;

	\path[coalition-red]
(P3)--(PY3)--(PX3)--cycle;

	\path[coalition-red]
(P4)--(PY4)--(PX4)--cycle;

	\path[coalition-red]
(P5)--(PY5)--(PX5)--cycle;

	\path[coalition-red]
(P6)--(PY6)--(PX6)--cycle;

	\draw[fill=black, radius=\dim] 
	(X) circle[radius=\dim]
	node [anchor=north, outer ysep=0.1cm] {$\alpha^t$} 
	(P1) circle[radius=\dim]
	node [anchor=north, outer ysep=0.1cm] {$\beta^t$}
	(P2) circle[radius=\dim]  
	(P3) circle[radius=\dim] 
	(P4) circle[radius=\dim] 
	(P5) circle[radius=\dim] 
	(P6) circle[radius=\dim];
	
	\draw (X) -- (P1);
	\draw (X) -- (P2);
	\draw (X) -- (P3);
	\draw (X) -- (P4);
	\draw (X) -- (P5);
	\draw (X) -- (P6);

	\draw[dashed] (P1) -- (PX1);
	\draw[dashed] (P2) -- (PX2);
	\draw[dashed] (P3) -- (PX3);
	\draw[dashed] (P4) -- (PX4);
 	\draw[dashed] (P5) -- (PX5);
	\draw[dashed] (P6) -- (PX6);
 	\draw[dashed] (P1) -- (PY1);
	\draw[dashed] (P2) -- (PY2);
 	\draw[dashed] (P3) -- (PY3);
	\draw[dashed] (P4) -- (PY4);
	\draw[dashed] (P5) -- (PY5);
	\draw[dashed] (P6) -- (PY6);

\end{tikzpicture}
    \caption{$\pi^{t+1}$}
    \label{fig:tree:t-step:B}
\end{subfigure}\caption{$t$-th steps of Algorithm \ref{alg:is:tree}.}
\label{fig:tree:t-step}
\end{figure}

  \begin{proof}[Proof of Claim \ref{tree:claim:one-edge}]
  For $t=0$ the claim is trivial.  
  Hence,  let us assume $t\geq 1$.  
  	If $i$ has not performed any deviation in the previous time steps $1, 2, \ldots, t-1$ then no edge in $L_{\pi^t}(i)$ has label equal to $i$. 
  	Otherwise, let $t' < t$ be the time step such that  $\alpha^{t'} = i$ and $\alpha^{t''} \neq i$, for every $t''$ such that  $t' < t'' < t$, i.e., $t'$ is the most recent time step in which $i$ has performed a deviation.  
   Then, by Claim \ref{tree:claim1}, $L_{\pi^{t'+1}}(i) = \{(i, \beta^{t'})\}$ and $\ell^{t'+1}((i, \beta^{t'})) = i$.
  	Since $i$ is not performing any deviation at any subsequent time step $t''$, $L_{\pi^{t''+1}}(i)$ will not contain any other edge with label $i$.  
  \end{proof}

\begin{claim}\label{tree:claim2}
	For every $i\in N$ and  $t\geq 0$ we have the  following:
	\begin{itemize}
		\item If $\alpha^t = i$  then $u_i(\pi^{t}) < u_i(\pi^{t+1})$ and $u_i(\pi^{t+1}) = v_i(\beta^t)$;
        \item If $\alpha^t \in N(i)$, $e=(\alpha^t,i) \in {\overline L}_{\pi^{t}}(i)$ and $i  = \beta^t$ 
  then $u_i(\pi^{t}) \leq  u_i(\pi^{t+1})$; 
		\item If $\alpha^t \in N(i)$ and $e=(\alpha^t,i) \in L_{\pi^{t}}(i)$ 
  then $u_i(\pi^{t}) \geq  u_i(\pi^{t+1})$; 
  \item If $\alpha^t \in N(i)$, $e=(\alpha^t,i) \in L_{\pi^{t}}(i)$ and  $\ell^t(e) = i$ 
  then $u_i(\pi^{t}) >  u_i(\pi^{t+1})$;
  \item otherwise $u_i(\pi^{t}) = u_i(\pi^{t+1})$. 
	\end{itemize}
\end{claim}

\begin{proof}[Proof of Lemma \ref{tree:lemma:join}]
The proof of the lemma trivially follows from the following claim.
\begin{claim}\label{tree:claim:join}
For every $i \in N$ and  $j\in N(i)$, 
\begin{align}
    |T_j(i)|
    &\leq 1 + \sum_{q \in N(i)} |T^{i}(q)|. \nonumber
\end{align}
\end{claim}
\begin{proof}    
We consider the times steps depicted in figure  \ref{fig:tree:sequence}. 
Let $t, t'$ be two consecutive times steps in $T_j(i)$, i.e., $t<t'$ and ${t''}  \notin T_j(i)$ for every $t < t'' < t'$. 
Notice that  $t' \neq t+1$. 
Since  $u_i(\pi^{t+1})  = u_i(\pi^{t'+1}) = v_i(j)$,  we  can deduce that there  must exist an intermediate time step in which the utility of $i$ has strictly decreased, and in  particular it must strictly decrease below $v_i(j)$. 
Let ${\bar t} + 1$ denote the first time step after $t$ such that $u_i(\pi^{{\bar t}+1})  <  v_i(j)$, i.e., 
$t +  1 < {\bar t}  + 1 \leq t'$ and, for every $t''$ such that $t + 1 \leq t'' \leq {\bar t}$, we have $u_i(\pi^{{t''}})  \geq  v_i(j)$.
By the  third and forth points of Claim \ref{tree:claim2} it must hold that at time $\bar t$ there must be some neighbour of $i$ that breaks  the  edge with $i$, i.e.,
$\alpha^{\bar t} \in N(i)$ and $(\alpha^{\bar t},i) \in L_{\pi^{\bar t}}(i)$.
We  want to show that there  must exists a  time  step $\tau$ such that $t+1 \leq \tau \leq \bar{t}$ and $\tau  \in T^{i}(\alpha^{\tau})$.   
Assume by contradiction that there is no such step $\tau$.
Then, by Claim \ref{tree:claim:one-edge}, in every time step $t''$ such that $t+1 \leq t'' \leq \bar{t} + 1$,  
there must always exist an edge in $L_{\pi^{t''}}(i)$ with label $i$. 
Let us denote by $e_{t''}$ the unique edge in $L_{\pi^{t''}}(i)$ with label $i$. 
Notice that $e_{t''}$ is either the edge $(i,j)$, built by $i$ at time $t$, or a new edge built by $i$ in a subsequent time step. 
Since, by the first point of Claim \ref{tree:claim2}, at every deviation of $i$ the utility of $i$ is strictly increasing, we must have that, if $e_{t''} = (i, j'')$, $v_i(j'') \geq v_i(j)$;
this implies $u_i(\pi^{t''}) \geq v_i(j'') \geq v_i(j)$ and hence $u_i(\pi^{\bar t  + 1}) \geq v_i(j)$. 
The existence of $\tau$ implies that $i$ can build the edge with $j\in N(i)$ a number of times that is at most one plus the number of times that any player in $N(i)$ can break the edge with $i$ when its label is equal to $i$.
\end{proof}

\end{proof}
\begin{figure}[h]
    \centering
        \begin{tikzpicture}[scale=0.95]

\def\start{0}
\def\fine{9}

\def\plus{\text{+}}

\def\t{1}
\def\tp{2}
\def\interA{3}
\def\tbar{4}
\def\tbarp{5}
\def\interB{6}
\def\tprime{7}
\def\tprimep{8}

	\coordinate (start) at (\start,0);
	\coordinate (t) at (\t,0);
	\coordinate (tp) at (\tp,0);
	\coordinate (interA) at (\interA,0);
	\coordinate (tbar) at (\tbar,0);
	\coordinate (tbarp) at (\tbarp,0);
	\coordinate (interB) at (\interB,0);
	\coordinate (tprime) at (\tprime,0);
	\coordinate (tprimep) at (\tprimep,0);

\draw[thick, ->] (\start, 0) -- (\fine, 0);

 Draw the ticks and labels
\foreach \x in {\t,\tp,\tbar,\tbarp,\tprime,\tprimep} {
    \draw[thick] (\x, 0.15) -- (\x, -0.15);
}

    \node at (\t, -0.4) {$t$};
    \node at (\tp, -0.4) {$t\plus 1$};
    \node at (\tp, +0.4) {$=v_i(j)$};
    \node at (\interA, -0.4) {$\leq$};
    \node at (\tbar, -0.4) {$\bar{t}$};
	\node at (\tbarp, -0.4) {$\bar{t}\plus 1$};
	\node at (\tbarp, +0.4) {$< v_i(j)$};
	\node at (\interB, -0.4) {$\leq$};
	\node at (\tprime, -0.4) {${t'}$};
	\node at (\tprimep, -0.4) {${t'}\plus 1$};
	\node at (\tprimep, +0.4) {$=v_i(j)$};
		
	\draw[dashed,<->] (\tp, -0.8) -- (\tbar, -0.8);
	\node at (\interA, -1) {$\tau$};
\end{tikzpicture}
    \caption{Change of the utility of player $i$.}
    \label{fig:tree:sequence}
\end{figure}
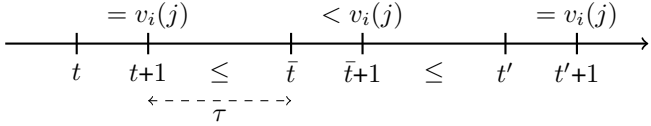

\begin{claim}\label{tree:claim:break}
For every $i\in N$, with $|N(i)|\geq 2$, and $j,j'\in N(i)$, with $j\neq j'$,  
	\begin{align}
	|T^{j'}_j(i)|
    &\leq 1 + \sum_{q \in N(i)\setminus \{j'\}} |T^{i}(q)|\nonumber.
\end{align}
\end{claim}
\begin{proof}
The proof of the claim is the proof of Claim \ref{tree:claim:join} with the additional observation that, due to the deviation of $i$ at time $t \in T^{j'}_j(i)$, we have $v_i(j') < v_i(j)$. 
This implies that before time $\tau$, $i$ never build an edge with $j'$ and hence $j'$ cannot  be the player breaking the edge labelled $i$ with $i$ at time $\tau$. 
\end{proof}

Claim \ref{tree:claim:break} implies the following key lemma. 
\begin{lemma}\label{tree:lemma:break}
For every $i\in N$, with $|N(i)|\geq 2$, $j'\in N(i)$ and $Q
 \subseteq N(i)\setminus \{j'\}$,
	\begin{align}
	|T^{j'}_{Q}(i)|
    &\leq |Q|\left(1 + \sum_{q \in N(i)\setminus \{j'\}} |T^{i}(q)|\right).\nonumber
\end{align}
\end{lemma}

\begin{proof}[Proof of Claim \ref{tree:moves:i}]
	Let $A_k = \{i\in N\setminus \{r\} : d_i \leq k\}$ for some  integer $k \in [0, d_r-1]$.
	We prove by induction on $k$ that the inequality  holds for every $i\in A_k$ and every $k$. 
	If $k=0$  then every $i\in A_k$ is a leaf which implies $C_i = \emptyset$ and $D_i=\{i\}$. Hence  both sides of the inequality are equal  to $0$. 
	Assume that the  claim is true for $A_k$, we prove the  claim for $A_{k+1}$.
	Let $i$ be any player in $A_{k+1}$.
 Since $C_i = N(i)\setminus \{p(i)\}$,
	from Lemma \ref{tree:lemma:break} we have 
	\begin{align}\label{tree:moves:step1}
    |T^{p(i)}_{C_i}(i)|
    &\leq |C_i|\big(1 + \sum_{q \in C_i
    } |T^{i}(q)|\big). 
\end{align}
Moreover, by definition, for every $q \in C_i$,
\begin{align}\label{tree:moves:step2}
	|T^{i}(q)| = |T^i_{C_q}(q)|.
\end{align}
Finally, by inductive hypothesis, we have that for every $q\in C_i$ \begin{align}\label{tree:moves:step3}
	|T^i_{C_q}(q)| \leq \sum_{j\in D_q} m^q_j. 
\end{align}

Hence,  
	\begin{align}
    |T^{p(i)}_{C_i}(i)|
    &\leq |C_i|\big(1 + \sum_{q \in C_i
    } \sum_{j\in D_q} m^q_j\big)\nonumber\\
    &= |C_i| + |C_i|\sum_{q \in C_i
    } \sum_{j\in D_q} m^q_j\nonumber\\
    &= m^i_i + \sum_{ q  \in C_i
    }  \sum_{j\in D_q} \left( |C_i| \cdot   m^q_j\right)\nonumber\\
    &= m^i_i + \sum_{ q  \in C_i
    } \sum_{j\in D_q} m^i_j\nonumber\\
    &=  
    \sum_{j\in D_i} m^i_j,\nonumber
\end{align}
where the first inequality follows by combining \eqref{tree:moves:step1},  \eqref{tree:moves:step2} and \eqref{tree:moves:step3}.   
\end{proof}

\bigskip
\subsubsection*{Proof of Corollary \ref{cor:tree}}
 \begin{proof}
We first focus on paths.
Notice that in the tree rooted in $r$ we have $|C_r| \leq 2$ and $|C_j| = 1$ for every $j\neq r$ which is not a leaf.
Hence, $m^r_j \leq 2$ for every $j\in D(r)$, which implies, by  Lemma \ref{tree:moves:r},  $|T(r)|\leq 2n$. 

We now prove the claim for stars and always apply Lemma \ref{tree:moves:r}.
If $r$ is the center of the star, we have $|T(r)| \leq m^r_r = |C_r| = n-1$. 
On the other hand, if $r$ is a leaf of the start, let $x\neq r$ be the center of the star, we have $|T(r)| \leq m^r_r + m^r_x = |C_r| + |C_x|\cdot|C_r| = 1 + |C_x| = 1 + (n-2) = n-1$.
\end{proof}

\bigskip
\subsubsection*{Proof of Proposition \ref{prop:exp_time}}
\begin{proof}
Take some positive integer $t$ and consider the instance depicted on Figure \ref{fig:tree_exp_moves}. There are $2t+1$ players named $\{x_1, \ldots, x_{x_{t+1}}\} \cup \{y_1, \ldots,y_t\}$. For $i=1$ to $t$, player $x_i$ has value $2$ for player $x_{i+1}$, and value $1$ for player $y_i$. In any other case, the value is 0. Since the instance is additively separable with non negative values, every request for joining a coalition is accepted. The proof relies on the following claim that is used inductively: 

{\em There exists a sequence of states where the number of deviations made by $x_i$ is at least two times the number of deviations made by  $x_{i+1},$ for every $i\in [t-1]$.}

In order to prove the claim, take some $i\in [t-1]$ and consider the following pattern inside the sequence: $(a)$ $x_i$ is neither  with $y_i$ nor  $x_{i+1}$ ($x_i$'s utility is 0), $(b)$ $x_i$ joins $\{y_i\}$ ($x_i$'s utility becomes 1), $(c)$ $x_i$ joins the coalition of $x_{i+1}$ ($x_i$'s utility becomes 2), and $(d)$ $x_{i+1}$ moves so that $x_i$ becomes alone ($x_i$'s utility becomes 0). Between two consecutive steps of the pattern of $i$, some other independent deviations can occur such as ``$x_i$ accepts that $x_{i-1}$ joins her coalition''. Such a step is also part of a pattern, but for an index different from $i$.  We consider a unique sequence where all the patterns, for  $i=1$ to $t-1$, appear multiple times,  and are entangled. Some steps are share by two patterns, but for different values of $i$. For example, step $(d)$ for index $i$ can also be step $(b)$ or step $(c)$ for index $i+1$. 
In the initial state of the sequence,  every player is alone (step $(a)$ of every pattern).

A pattern of $i$ includes two deviations by $x_i$ (steps $(b)$ and $(c)$) which can be repeated as many times as $x_{i+1}$ ``abandons'' $x_i$ (step   $(d)$).   
In order to see this, let us describe a part of the unique sequence from the viewpoint of $x_i$ where $i \in [t-1]$.  
In the first state, player $x_i$ is neither with $y_1$ nor $x_{i+1}$ (step $(a)$ of pattern $i$). 
Afterwards, player $x_i$ joins $y_i$ (step $(b)$ of pattern $i$).  
Then, player $x_i$ joins $x_{i+1}$ (step $(c)$ of pattern $i$). 
Then, player $x_{i+1}$ joins $y_{i+1}$ (step $(d)$ of pattern $i$ and step $(b)$ of pattern $i+1$ as well).  
Then, player $x_i$ joins $y_i$ (second occurrence of step $(b)$ of pattern $i$).  
Then, player $x_i$ joins $\{x_{i+1},y_{i+1}\}$ (second occurrence of step $(c)$ of pattern $i$).  
Then, player $x_{i+1}$ joins $\{x_{i+2}\}$ (second occurrence of step $(d)$ of pattern $i$ and step $(c)$ of pattern $i+1$).  
Then, player $x_i$ joins $y_i$ (third occurrence of step $(b)$ of pattern $i$), and player $x_i$ joins $\{x_{i+1},x_{i+2}\}$ (third occurrence of step $(c)$ of pattern $i$). So, we have described six deviations by $x_i$ (namely, three steps $(b)$ and three steps $(c)$), and four of them follow the two defections by $x_{i+1}$ (two steps $(d)$). Therefore, if $x_{i+1}$ can ``abandon'' $x_{i}$ $\delta$ times, then $x_{i}$ can do $2+2\delta$ moves.

The instance is such that $x_i$ is affected by the deviations of $x_{i+1}$ but 
$x_{i+1}$ is not affected by the deviations of $x_i$:  
$x_{i+1}$ simply accepts each time $x_i$ decides to join her coalition, and the utility of $x_{i+1}$ remains unchanged in that case. That is why  the patterns for $i=1$ to $t-1$ can occur in the same sequence.

The number of moves in the sequence can be lower bounded by induction. The induction basis depends on the number of moves of $x_{t}$. Then, from $i=t-1$ down to $1$, we can estimate the number of moves made by $x_i$.

Starting from an initial state where all the players are alone, player $x_t$ can make two deviations: joining $\{y_t\}$, and then joining $\{x_{t+1}\}$. Based on this, the claim tells us that $x_{t-1}$ makes at least $2^2$ deviations.  More generally, player $x_i$ makes at least $2^\Delta$ deviations where $i=t+1-\Delta$ and $i\in [t]$. 
In total, there are at least $\sum_{i=1}^t 2^\Delta=\sum_{i=1}^t 2^{t+1-i}=\sum_{i=1}^t 2^{i}=2^{t+1}-2$ deviations in the proposed sequence. Since the number of players $n$ is equal to $2t+1$, we get that the proposed sequence requires at least $2^{(n+1)/2}-2$ deviations, which is exponential in $n$.

Lastly, notice that the number of deviations of node $x_1$ in the instance considered in Proposition  \ref{prop:exp_time} matches the bound of Lemma \ref{tree:moves:r}; in fact
$m^{x_1}_{y_k} = 0$, for $k \in \{1,2,\ldots,t\}$, $m^{x_1}_{x_{t+1}} = 0$  and $m^{x_1}_{x_t} = 2$, $m^{x_1}_{x_{t-1}} = 2\cdot 2$ and,  in general,  $m^{x_1}_{x_{t-k}} = 2^{k+1}$, hence $\sum_{k=0}^{t-1} m^{x_1}_{x_{t-k}} = \sum_{k=0}^{t-1} 2^{k+1} = 2^{t+1}-2$. 
\end{proof}

\subsubsection*{Quadratic convergence on paths}

\begin{example} \label{n_square_path}
$(N,L)$ is a path with $n$ players having LAS preferences. The players are denoted by $1$ to $n$ from left to right. Every player $i \in \{2,\ldots,n\}$ has value 1 for player $i-1$, and every player $i \in \{1,\ldots,n-1\}$ has value 2 for player $i+1$. Algorithm \ref{alg:is:lowerb} describes a possible execution of the IS dynamics for the instance. At each deviation, player $i$ joins the coalition of player $i+1$ for which she has value 2 (the previous utility was 0, if player $i$ were alone, or 1 if player $i$ were in the same coalition as $i-1$). By doing so, player $i \in \{2,n-1\}$ breaks the coalition that player  $i-1$ has created with her.  For each $r \in [n-1]$ we have $n-r$ deviations, i.e.,  $n(n-1)/2$ deviations in total.       
 
\begin{algorithm}                      
\caption{IS dynamics on a path of length $n$ (LAS preferences)}         
\label{alg:is:lowerb}                          
\begin{algorithmic}[1]                  
\STATE Every player is alone in the initial partition 
\FOR{$r=1$ {\bf to} $n-1$}
\FOR{$i=1$ {\bf to} $n-r$}
\STATE player $i$ deviates to the coalition of player $i+1$ 
\ENDFOR

\ENDFOR
 \end{algorithmic}
\end{algorithm}

\end{example}

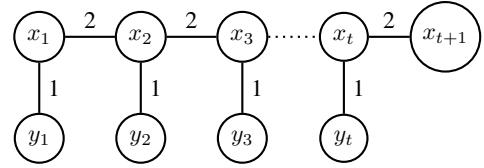
\begin{figure}

\begin{center}
\begin{tikzpicture}[node distance={15mm}, thick, main/.style = {draw, circle},
scale=0.9, transform shape] 
\node[main] (1) {$x_1$}; 
\node[main] (2) [right of=1]  {$x_2$}; 
\node[main] (3) [right of=2] {$x_3$}; 
\node[main] (4) [right of=3] {$x_{t}$}; 
\node[main] (5) [right of=4] {$x_{t+1}$}; 
\node[main] (6) [below of=1] {$y_1$}; 
\node[main] (7) [below of=2] {$y_2$}; 
\node[main] (8) [below of=3] {$y_3$}; 
\node[main] (9) [below of=4] {$y_{t}$};

\draw[dotted] (3) -- (4);

\draw[-] (1) -- node[midway, above, sloped, pos=0.5] {2} (2); 
\draw[-] (2) -- node[midway, above, sloped, pos=0.5] {2} (3); 
\draw[-] (4) -- node[midway, above, sloped, pos=0.5] {2} (5); 
\draw[-] (1) -- node[midway, right, pos=0.5] {1} (6); 
\draw[-] (2) -- node[midway, right, pos=0.5] {1} (7); 
\draw[-] (3) -- node[midway, right, pos=0.5] {1} (8); 
\draw[-] (4) -- node[midway, right, pos=0.5] {1} (9); 

\end{tikzpicture} 
\end{center}
\caption{Instance where the IS dynamics can converge after an exponential number of steps.}
\label{fig:tree_exp_moves}
\end{figure}

\end{document}